\documentclass{article}
\usepackage[utf8]{inputenc}
\usepackage{fec}
\usepackage{authblk}
\usepackage{bm}
\usepackage{subcaption}
\usepackage{mathtools}
\usepackage{enumitem}
\usepackage{amsmath}
\usepackage{amsfonts}
\usepackage{amsthm}
\usepackage{multirow}
\usepackage{algorithm}
\usepackage{algorithmic}
\usepackage{braket}
\usepackage{booktabs}
\usepackage{fullpage}

\newtheorem{theorem}{Theorem}[section]

\newtheorem{lemma}[theorem]{Lemma}

\usepackage{hyperref}

\DeclareMathOperator*{\argmin}{arg\,min}
\newcommand{\eg}{\epsilon_{g}}
\newcommand{\ez}{\epsilon_{m0}}
\newcommand{\eo}{\epsilon_{m1}}
\newcommand{\ezi}{\epsilon_{m0,i}}
\newcommand{\eoi}{\epsilon_{m1,i}}
\newcommand{\egi}{\epsilon_{g,i}}
\renewcommand{\Pr}{\text{Pr}}

\newcommand{\rhohat}{\hat{\rho}}
\newcommand{\rhotilde}{\tilde{\rho}}
\DeclareUnicodeCharacter{0301}{\'{e}}

\title{A Bayesian Approach for Characterizing and Mitigating Gate and Measurement Errors}
\author[1]{Muqing Zheng}
\author[2]{Ang Li}
\author[1]{Tam{\'a}s Terlaky}
\author[1]{Xiu Yang\thanks{Email: xiy518@lehigh.edu}}
\affil[1]{Department of Industrial and Systems Engineering, Lehigh University, Bethlehem, PA 18015}
\affil[2]{Advanced Computing, Mathematics and Data Division, Pacific Northwest National Laboratory, Richland, WA 99354}
\date{}

\begin{document}

\maketitle

\begin{abstract}
    Various noise models have been developed in quantum computing study to 
    describe the propagation and effect of the noise which is caused by imperfect implementation of hardware. 
    %In these models, critical parameters, e.g., readout error rates, are typically modeled as constants. 
    Identifying parameters such as gate and readout error rates are critical to these models.
    %Instead, we model such parameters as random variables, and apply a new Bayesian inference algorithm to classical gate and measurement error models to identify the distribution of these parameters. 
    We use a Bayesian inference approach to identify posterior distributions of these parameters, such that they can be characterized more elaborately. 
    By characterising the device errors in this way, we can further improve the accuracy of quantum error mitigation. 
    Experiments conducted on IBM's quantum computing devices suggest that our
    approach provides better error mitigation performance than existing
    techniques used by the vendor. %, in which error rates are estimated as deterministic values.
    Also, our approach outperforms the standard Bayesian inference method in
    some scenarios. 
%    {\color{blue}More importantly, our approach only requires data from $\Ocal(1)$ number of circuits. }

\end{abstract}

%===============================================================================

\section{Introduction}

While quantum computing (QC) displays an exciting potential in reducing the time complexity of various problems, the noise from environment and hardware may still undermine the advantages of QC algorithms~\cite{Preskill2018}. 
One of the solutions to this problem is quantum error correction
% (QEC)~\cite{Steane1998, Bravyi1998QuantumCO, Raussendorf2007, colorABC, Aharonov2008, Flhmann2019, Grimsmo2020, Joshi2020}, 
(QEC)~\cite{Raussendorf2007, colorABC, Aharonov2008, Flhmann2019, Grimsmo2020, Joshi2020}, 
%Usually, single-qubit information is encoded by an array of physical qubits so that the information stored in this single ``logic qubit'' is robust against errors through redundancy if the error rate is below a certain threshold~\cite{Aharonov2008}. 
%Usually, an array of physical qubits is encoded as a single ``logic qubit'' to against errors through redundancy, if the error rate is below a certain threshold~\cite{Aharonov2008}. 
which utilizes redundancy to protect the information of a single ``logic qubit'' from errors.
Two representative examples are surface code and color code due to their scalability and high error thresholds~\cite{Raussendorf2007, colorABC}. 
An alternative approach to QEC is bosonic codes. In this coding scheme, the
single-qubit information is encoded into a higher-dimensional system, like a
harmonic oscillator. One advantage of Bosonic codes is that it provides an
access to larger Hilbert space with less overhead than traditional QEC
codes~\cite{Flhmann2019, Grimsmo2020, Joshi2020}.

% {\color{blue} May need to add description of error correction code.}

However, as described in~\cite{Preskill2018}, in the ``noisy intermediate-scale
quantum (NISQ)'' era, the small- or medium-sized but noisy quantum computers cannot afford 
the cost of QEC codes because they impose a heavy overhead cost in number of qubits and number of gates. 
As a result, quantum error mitigation (QEM) techniques have become attractive,
e.g.,~\cite{Wallman2016, Temme2017, Endo2018, kandala2019error, QDT2020,
Bravyi2020, Arute2020,Chen2019, Geller2020,funcke2022measurement}, since their cost is much lower than the QEC codes in terms of the circuit depth and the number of qubits.
%In the noisy intermediate-scale quantum (NISQ) era~\cite{Preskill2018}, error 
%mitigation techniques have been attracted much attention,
%e.g.,~\cite{Fiurek2001, Temme2017, Endo2018,kandala2019error, Geller2020, QDT2020, Bravyi2020},
%as their cost is much lower than the quantum error correction codes in terms of circuit depth and number of qubits.
One important area in the error mitigation study is to filter the measurement errors (or readout errors). 
These errors are usually modeled by multiplying a stochastic matrix with a
probability vector, as such to depict the influence of the noise on the output of QC algorithms.
More precisely, the probability vector represents the desired noiseless
output of a QC algorithm, the stochastic matrix describes how the noise affects this
output, and the resulting vector consists of the probabilities of observing each
possible state on the quantum device.
Here, the stochastic matrix can be constructed from conditional probabilities if only classical errors are considered, or from results of tomography if non-classical errors are not
%also
significant~\cite{Chen2019, Geller2020, QDT2020, Bravyi2020}. 
Similarly, the study in~\cite{Takahashi2020} shows the possibility to simulate
bit-flip gate error in some quantum circuits in a classical manner. 

The goal of QEM from the algorithmic perspective is to recover the noise-free information using data from repeated experiments, which is usually achieved via statistical methods.
In the existing error models, the parameters, e.g., error rate of measurement or
gates, are usually considered as deterministic values (possibly with confidence
interval), and the goal is to filter the error in estimating the expectation of an operator.
%instead of some random variables according to some unknown distributions.  We did not see the related work for the latter case. 
Instead, by considering error mitigation as a stochastic inverse problem, we
adopt a new Bayesian algorithm from~\cite{Butler2018} to construct the
distributions of model parameters and use corresponding backward error models to
filter errors from the outcomes of a quantum device.
%gate errors and measurement errors from the measurement outcomes of quantum circuits.
Note that our framework does not rely on the specific knowledge of the problems that quantum circuits want to answer, like in~\cite{Arute2020}, or hardware calibration, such as~\cite{Sheldon2016, Arute2019}. 
We aim to estimate the parameters more comprehensively for selected error models
as an inverse problem while error mitigation is achieved by using the error model in a backward direction.
%{\color{blue} Need a short description of gate errors. A couple of sentences describing the we are focusing on the algorithm side, not the protocols for device, e.g., nature paper.  We focus on pure state measurements in this study.  }

The paper is organized as following. 
%To start with, we present the fundamental quantum computation model in Section~\ref{sec:background}. 
In Section~\ref{sec:error-models}, we provide the measurement error model based on independent classical measurement error and expand the gate error model in~\cite{Takahashi2020} to multiple-error scenario. 
In Section~\ref{sec:est}, we introduce the use of Bayesian algorithm in~\cite{Butler2018} to infer the distributions of parameters of measurement error and gate error models. 
Then, we demonstrate the creation of our error filter on IBM's quantum device \texttt{ibmqx2} (Yorktown) and apply our filter together with other existing error mitigation methods on measurement outcome from state tomography, an example of Grover's search~\cite{algs}, an instance of Quantum Approximate Optimization Algorithm (QAOA)~\cite{algs}, and a 200-NOT-gate circuit in Section~\ref{sec:exp}.
The code is available in~\cite{muqing_zheng_2022_7005312}.

\section{Error Models}\label{sec:error-models}

% \subsection{Notations}
The goals of our error modeling include estimating the influence of bit-flip gate errors and measurement errors in the outputs of a quantum circuit without accessing any quantum device and recovering the error-free (or error-mitigated) output. Throughout this paper, we assume no state-preparation error and only focus on pure state measurements. The three error rates that we care about are as follows:
%using some error parameters from a quantum computer. 
%In this study, we assume that each gate has independent depolarizing error and each measurement operator has independent bit-flip error. 
%In this study, we assume that each gate and measurement operator has independent bit-flip error. 
%We use three parameters to characterize the probabilities of the occurrence of different types of errors as follows:
\begin{enumerate}[itemsep=0pt, topsep=2pt]
    %\item $\epsilon_{g}$ = the chance of having depolarizing error in a gate, where $\eg \in (0,1)$;
    \item $\epsilon_{g}$ = the chance of having a bit-flip error in a gate;
    \item $\epsilon_{m0}$ = the chance of having a measurement error when measure $\ket{0}$;
    \item $\epsilon_{m1}$ = the chance of having a measurement error when measure $\ket{1}$.
\end{enumerate}
It is reasonable to consider $\eg \neq 0.5$ and $\ez + \eo \neq 1$ in the current quantum computer~\cite{Bravyi2020,Arute2019}. This assumption is one of the necessary conditions for the existence of the error-mitigation solutions in our following models.

%The depolarizing channel is a commonly assumed gate error model in practice~\cite{Arute2019}. 
%It assumes the equal probability of having bit flip, phase flip, and the combination of both.
%In 1-qubit case, the occurrence of depolarizing error is equivalent to the scenario that the qubit has $1 - \frac{3}{4}\epsilon_g$ chance to be error-free and $\frac{1}{4}\epsilon_g$ chance to have only bit-flip or phase-flip error, or the combination of both~\cite[p. 379]{NC2010}. 
%We use Pauli matrices $X$, $Z$ and $Y$ to characterize the bit-flip error, the phase-flip error, and the combination of both bit-flip and phase-flip error, respectively. 
%The depolarizing error rate $\epsilon_{g}$ build a connection between the probability of the occurrence of a general gate error and that of a bit-flip error. 
%The approach of constructing the gate error model using $\epsilon_g$ in Section~\ref{sec:bit-flip} is also applicable to other types of error models if such connection exists. 

\subsection{Measurement Error} \label{sec:me}
%\subsubsection{Independent Measurement Errors}  % too much subsection titles
%We assume that measurement errors are independent across qubits.
As is demonstrated in~\cite{Chen2019}, classical measurement error is applicable in the device we conduct experiments on, i.e., \texttt{ibmqx2}. 
We build measurement error model using conditional probabilities.
Consider a single-qubit state $\alpha \ket{0} + \beta \ket{1}$, its distribution of the noisy measurement outcomes are
% with given $\epsilon_{m0}$ and $\epsilon_{m1}$, we have
\begin{align*}
    \text{Pr}(\text{Measure 0 w/ noise}) &= |\alpha|^2 \cdot (1 - \epsilon_{m0}) + |\beta|^2 \cdot \epsilon_{m1} \\
    \text{Pr}(\text{Measure 1 w/ noise}) &= |\alpha|^2 \cdot \epsilon_{m0} + |\beta|^2 \cdot (1-\epsilon_{m1}),
\end{align*}
which is equivalent to
\begin{equation}
  \label{eq:error_mat0}
\small{
    \begin{bmatrix}
    1 - \ez & \eo \\
    \ez & 1 - \eo
    \end{bmatrix}
    \begin{pmatrix}
    \Pr(\text{Measure 0 w/o noise}) \\
    \Pr(\text{Measure 1 w/o noise})
    \end{pmatrix}
    =
    \begin{pmatrix}
    \Pr(\text{Measure 0 w/ noise}) \\
    \Pr(\text{Measure 1 w/ noise})
    \end{pmatrix},
}    
\end{equation}
where ``w/'' stands for ``with'' and ``w/o'' stands for ``without.''
Denoting $\ez$ and $\eo$ for qubit $i$ as $\ezi$ and $\eoi$, respectively, 
we can extend the matrix form in Eq.~\eqref{eq:error_mat0} to an $n$-qubit case
\begin{equation} \label{eq:me}
Ar = \rtilde, 
\end{equation}
where 
\begin{align*}
A &:= \bigotimes_{i = 1}^n\begin{bmatrix}1 - \ezi & \eoi \\ \ezi & 1 - \eoi\end{bmatrix}, \\
r &:=     \begin{pmatrix}
    \Pr(\text{Measure 0...00 w/o noise}) \\
    \Pr(\text{Measure 0...01 w/o noise}) \\
    \vdots \\
    \Pr(\text{Measure 1...11 w/o noise})
    \end{pmatrix}, \\
\rtilde &:=     \begin{pmatrix}
    \Pr(\text{Measure 0...00 w/ noise}) \\
    \Pr(\text{Measure 0...01 w/ noise}) \\
    \vdots \\
    \Pr(\text{Measure 1...11 w/ noise})
    \end{pmatrix},
 \end{align*}
%$A \in [0,1]^{2^n \times 2^n}$, $r \in [0,1]^{2^n}$, $\rtilde \in [0,1]^{2^n}$. 
$A_{ij}\in [0,1], r_i\in [0,1], \tilde r_i\in [0,1]$ by introducing the \textit{independence of measurement errors across qubits}.
We aim to identify $r$, but, in practice, we only have $\rtilde$ which is the probability vector characterizing the observed results from repeated measurements.
%In practice, $\rtilde$ is the (predictive) probability vector from experiment observation depends on the purpose of experiment (prediction or denoising). 
Note that $A$ is a nonnegative left stochastic matrix (i.e., each column sums to 1), so if $r \geq 0$ and its entries sums to 1, $\rtilde \geq 0$ and its entries also sums to 1.%}}}

%\subsubsection{Measurement Error Filter}\label{sec:meas-filter}
If $\ezi$ and $\eoi$ for all $i = 1,...,n$ are known, the most straightforward denoising method derived from Eq.~\eqref{eq:me} is $r:=A^{-1}\rtilde$.
%where $r^*$ is the probability vector without measurement error. 
As $\ezi + \eoi \neq  1$ for all $i = 1,...,n$,
%$i = 1,...,2^n$
each individual 2-by-2 matrix has non-zero determinant. 
Thus $A$ has non-zero determinant and $A^{-1}$ exists. 
However, it is not guaranteed that $r^*$ is a valid probability vector. An alternative is to find a constrained approximation
\begin{equation} \label{eq:me-filter}
%    r^*:= \argmin_{ \sum_{i = 1}^{2^n} r_i = 1, \forall i\in {1,...,2^n}\ r_i \geq 0} \| A^{-1}\rtilde - r  \|_2,
   r^*:= \argmin_{ \sum_{i = 1}^{2^n} r_i = 1, \forall i\in \{1,...,2^n\}\ r_i \geq 0} \| Ar - \rtilde \|_2.
\end{equation}
% and use $r^*$ to approximate $r$.

\subsection{Bit-flip Gate Error} \label{sec:bit-flip}

% We \textit{only consider} the situation when phase-flip error does not affect a QC
% algorithm's outcomes. So the influence of  
% In this case, $Y$ error is equivalent to $X$ error. 
% Together with the original bit-flip error, for a 1-qubit gate with depolarizing error rate $\epsilon_{g}$, the probability of the occurrence of a bit-flip error is $\frac{1}{2}\epsilon_{g}$~\cite{NC2010}. 
% Moreover, we will extend the error modeling approach to an $n-$qubit gate case, and we will demonstrate that $\frac{1}{2}\epsilon_{g}$ is still our major focus as in the 1-qubit scenario. 
% The factor $\frac{1}{2}$ was introduced in the original literature~\cite{Takahashi2020}, and it will not affect the validity of our conclusion.
%Although the rest of proof in this section is general for an $n-$qubit gate, we remain using $\frac{1}{2}\epsilon_{g}$ as our major focus is 1-qubit scenario. The factor $\frac{1}{2}$ will not affect the validity of the proof.

%Figure~\ref{fig:bf} illustrates a typical circuit we investigate using a single bit-flip error
%model. Here, $U$ is a noisy-free unitary gate and $X_{\epsilon_{g}/2}$ represents a bit-flip error with probability $\frac{1}{2}\epsilon_{g}$.

%The following error model is from \cite{Takahashi2020}. 
For the gate error, we focus on the single bit-flip error in this work, and we adopt the error model proposed in~\cite{Takahashi2020}.
%We adopt the error model into describe the single bit-flip error.
Of note, there is no direct proof in~\cite{Takahashi2020} to validate this model.
% Here, we complete the proof and also extend it to a multiple-error case in Section~\ref{sec:fem}. 
In this section, we complete the proof and also extend it to a multiple-error case.
% Using this method, we first consider the case when there is only one gate and one bit-flip error (on all qubits) in the circuit, as shown in Figure~\ref{fig:bf}. 
We first consider the case when there is only one gate and qubits could have bit-flip errors (all in the same rate $\eg$) after this gate, as shown in Figure~\ref{fig:bf}. 
\begin{figure}
    \centering
    % \[
    % \Qcircuit @C=1em @R=.7em {
    % \lstick{} &\qw &\multigate{2}{U}     &\gate{X_{\epsilon_{g}}} &\meter \\
    %   &\lstick{\vdots}  &\nghost{U}        &\vdots                       &\vdots \\
    % \lstick{} &\qw &\ghost{U}        &\gate{X_{\epsilon_{g}}}        &\meter \inputgroupv{1}{3}{0.75em}{2.2em}{\ket{\phi}}
    % }
    % \]
    \includegraphics[width=0.3\linewidth]{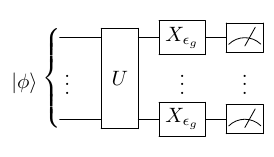}
    \caption{Single bit-flip error model. $U$ is a noisy-free gate and $X_{\epsilon_{g}}$ represents a bit-flip error with probability $\epsilon_{g}$.\label{fig:bf}}
\end{figure}

\subsubsection{Single Bit-flip Error} \label{sec:FE}

Let $p: \{0,1\}^n \rightarrow [0,1]$, where $n$ is the number of qubits, be the
Boolean function that represents the noise-free probability distribution of the
outcome of a QC algorithm and $x \in \{0,1\}^n$ denote the basis used in a QC algorithm. 
% We perform Fourier expansion of this Boolean function as
The Fourier expansion of this Boolean function is
\begin{equation}
    p(x) = \sum_{s \in \{0,1\}^n} \hat{p}(s)(-1)^{s.x} \label{eq:pFE},
\end{equation}
where $\hat{p}(s)$ is the Fourier coefficient of $p$ and $s.x = \sum_{i = 1}^ns_i\cdot x_i$ \cite[p.22]{anaBF}. 
These Fourier coefficients can be computed from
\begin{displaymath}
    \hat{p}(s) = \frac{1}{2^n}\sum_{x \in \{0,1\}^n} p(x) (-1)^{s.x}.
\end{displaymath}
%Let $y$ be the erroneous version of $x$ induced by the bit-flip error, i.e.,
Let $y$ be the erroneous version of $x$ induced by the bit-flip error. In other words, $y$ is a function of $x$ that adds bit-flip error into the measurement outcomes. The mathematical expression of $y$ is
\begin{displaymath} 
    y_i = \begin{cases}
    x_i &\text{ with probability $1 - \eg$} \\
    \neg x_i &\text{ with probability $\eg$}
    \end{cases} \text{ for $i = 1,...,n$}.
\end{displaymath}
Define $\ptilde:\{0,1\}^n \rightarrow [0,1]$ to be the expected distribution function of measurement outcomes under the noise model.
Then Eq.~\eqref{eq:pFE} implies
% \begin{displaymath}
%     \ptilde(x) = E_{y}[p(y)] = \sum_{s \in \{0,1\}^n} \hat{p}(s)E_{y}[(-1)^{s.x} ].
% \end{displaymath} 
\begin{displaymath} % Muqing: I don't feel original version is right.
    \ptilde(x) = E_{x}[p(y)] = \sum_{s \in \{0,1\}^n} \hat{p}(s)E_{x}[(-1)^{s.y} ].
\end{displaymath}
%Let us now only look at $E_{y}[(-1)^{s.x} ]$. 
It is clear that
% \begin{equation}
%     \label{eq:ptunf}
% \begin{aligned}
%     E_{y}[(-1)^{s.x} ] &= E_{y}\left[\prod_{i=1}^n(-1)^{s_i \cdot y_i}\right] \\
%     &= \prod_{i=1}^nE_{y}\left[(-1)^{s_i \cdot y_i}\right] \\
%     &= \prod_{i=1}^n\left[\left(1- \frac{\epsilon_{g}}{2}\right)(-1)^{s_i \cdot x_i} + \frac{\epsilon_{g}}{2}(-1)^{s_i \cdot \neg x_i}\right].
% \end{aligned}    
% \end{equation}
\begin{equation}
    \label{eq:ptunf}
\begin{aligned}
    E_{x}[(-1)^{s.y} ] &= E_{y}\left[\prod_{i=1}^n(-1)^{s_i \cdot y_i}\right] \\
    &= \prod_{i=1}^nE_{y}\left[(-1)^{s_i \cdot y_i}\right] \\
    &= \prod_{i=1}^n\left[\left(1- \eg\right)(-1)^{s_i \cdot x_i} + \eg(-1)^{s_i \cdot \neg x_i}\right].
\end{aligned}    
\end{equation}
Since $x_i$ and $s_i$ are binary bits, there are four possible cases:
% \begin{itemize}[itemsep=0pt, topsep=2pt]
%     \item $s_i = 0,x_i = 0$, then $\left(1- \frac{\epsilon_{g}}{2}\right)(-1)^{s_i \cdot x_i} + \ \frac{\epsilon_{g}}{2}(-1)^{s_i \cdot \neg x_i} = 1$;
%     \item $s_i = 0,x_i = 1$, then $\left(1- \frac{\epsilon_{g}}{2}\right)(-1)^{s_i \cdot x_i} + \ \frac{\epsilon_{g}}{2}(-1)^{s_i \cdot \neg x_i} = 1$;
%     \item $s_i = 1,x_i = 0$, then $\left(1- \frac{\epsilon_{g}}{2}\right)(-1)^{s_i \cdot x_i} + \ \frac{\epsilon_{g}}{2}(-1)^{s_i \cdot \neg x_i} = 1-\epsilon_{g}$;
%     \item $s_i = 1,x_i = 1$, then $\left(1- \frac{\epsilon_{g}}{2}\right)(-1)^{s_i \cdot x_i} + \ \frac{\epsilon_{g}}{2}(-1)^{s_i \cdot \neg x_i} = (1-\epsilon_{g})\cdot (-1)$.
% \end{itemize}
\begin{itemize}[itemsep=0pt, topsep=2pt]
    \item $s_i = 0,x_i = 0$, then $\left(1- \eg\right)(-1)^{s_i \cdot x_i} + \eg(-1)^{s_i \cdot \neg x_i} = 1$;
    \item $s_i = 0,x_i = 1$, then $\left(1- \eg\right)(-1)^{s_i \cdot x_i} + \eg(-1)^{s_i \cdot \neg x_i} = 1$;
    \item $s_i = 1,x_i = 0$, then $\left(1- \eg\right)(-1)^{s_i \cdot x_i} + \eg(-1)^{s_i \cdot \neg x_i} = 1-2\epsilon_{g}$;
    \item $s_i = 1,x_i = 1$, then $\left(1- \eg\right)(-1)^{s_i \cdot x_i} + \eg(-1)^{s_i \cdot \neg x_i} = (1-2\epsilon_{g})\cdot (-1)$.
\end{itemize}
To summarize,
% \begin{equation*} 
%     \left(1- \frac{\epsilon_{g}}{2}\right)(-1)^{s_i \cdot x_i} + \ \frac{\epsilon_{g}}{2}(-1)^{s_i \cdot \neg x_i} = (1-\epsilon_{g})^{s_i}(-1)^{s_i \cdot x_i},
% \end{equation*}
\begin{equation*} 
    \left(1- \eg\right)(-1)^{s_i \cdot x_i} + \eg(-1)^{s_i \cdot \neg x_i} = (1-2\epsilon_{g})^{s_i}(-1)^{s_i \cdot x_i},
\end{equation*}
for all $s_i \in \{0,1\}$ and $x_i \in \{0,1\}$.
% Consequently, based on Eq.~\eqref{eq:ptunf}, we have
% \begin{displaymath}
%     E_{y}[(-1)^{s.x} ] &= \prod_{i=1}^n\left[\left(1- \frac{\epsilon_{g}}{2}\right)(-1)^{s_i \cdot x_i} + \ \frac{\epsilon_{g}}{2}(-1)^{s_i \cdot \neg x_i}\right] \\
%     &= \prod_{i=1}^n\left[(1-\epsilon_{g})^{s_i}(-1)^{s_i \cdot x_i}\right]\\
%     &= (1-\epsilon_{g})^{|s|} (-1)^{s.x},
% \end{displaymath}
Consequently, continuing from Eq.~\eqref{eq:ptunf},
\begin{displaymath}
    E_{y}[(-1)^{s.x} ] =  \prod_{i=1}^n\left[(1-2\epsilon_{g})^{s_i}(-1)^{s_i \cdot x_i}\right] = (1-2\epsilon_{g})^{|s|} (-1)^{s.x},
\end{displaymath}
where $|s| = \sum_{i = 1}^n s_i$. Thus, the $\ptilde$ with only one bit-flip error is
% \begin{displaymath}
%     \Tilde{p}(x) = E_y[p(y)] = \sum_{s \in \{0,1\}^n}  (1-\epsilon_{g})^{|s|}\hat{p}(s)(-1)^{s.x}.
% \end{displaymath}
\begin{displaymath}
    \Tilde{p}(x) = E_x[p(y)] = \sum_{s \in \{0,1\}^n}  (1-2\epsilon_{g})^{|s|}\hat{p}(s)(-1)^{s.x}.
\end{displaymath}

% Muqing: I re-organize the entire subsection, so I comment out all old content
\subsubsection{Extension to Multiple Bit-flip Errors} \label{sec:fem}

The extension is only applicable on gates that \textit{commute} with $X$ gate up to a global phase factor. This condition allows us to move occurred bit-flip errors to the end of the circuit, like the change from Figure~\ref{fig:bfmp} to~\ref{fig:bfm}, where $U_1,...,U_m$ are still noisy-free unitary gates. The model is constructed by repeatedly apply the previous proof procedure, instead of considering the cancellation of errors, since our interest is on individual gates but not on the accumulated one. 
\begin{figure}
     \centering
     \begin{subfigure}[b]{0.49\textwidth}
         \centering
        % \[
        % \Qcircuit @C=0.7em @R=.7em {
        % &\qw             &\multigate{2}{U_1}  &\gate{X_{\epsilon_{g}}} &\qw &\cdots& &\multigate{2}{U_m}&\gate{X_{\epsilon_{g}}} &\meter \\
        % &\lstick{\vdots} & \nghost{U_1}       &\vdots                  &    &\cdots& &\nghost{U_m}      &\vdots                  &\vdots \\
        % &\qw             &\ghost{U_1}         &\gate{X_{\epsilon_{g}}} &\qw &\cdots& &\ghost{U_m}       &\gate{X_{\epsilon_{g}}} &\meter
        % \inputgroupv{1}{3}{0.75em}{2em}{\ket{\phi}}   \gategroup{1}{8}{3}{9}{1em}{--} \gategroup{1}{3}{3}{4}{1em}{--}\\
        % &\\
        % & & & & &\mbox{$m$ times}
        % }
        % \]
        \includegraphics[width=\linewidth]{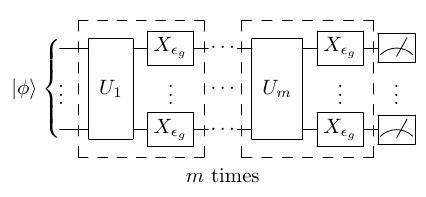}
        \caption{\footnotesize A practical multi-gate bit-flip error model\label{fig:bfmp}}
     \end{subfigure}
     \hfill
     \begin{subfigure}[b]{0.49\textwidth}
         \centering
        % \[
        % \Qcircuit @C=0.7em @R=.7em {
        % &\qw              &\multigate{2}{U_1}  &\qw &\cdots& &\multigate{2}{U_m} &\gate{X_{\epsilon_{g}}} &\qw &\cdots&  &\gate{X_{\epsilon_{g}}} &\meter \\
        % &\lstick{\vdots}  & \nghost{U_1}       &    &\cdots& &\nghost{U_m}       &\vdots       &    &\cdots&  &\vdots       &\vdots \\
        % &\qw              &\ghost{U_1}         &\qw &\cdots& &\ghost{U_m}        &\gate{X_{\epsilon_{g}}}        &\qw &\cdots&  &\gate{X_{\epsilon_{g}}}        &\meter \inputgroupv{1}{3}{0.75em}{2em}{\ket{\phi}}   \gategroup{1}{8}{3}{8}{1em}{--} \gategroup{1}{12}{3}{12}{1em}{--} \\
        % &\\
        % & & & & & & & & &\mbox{$m$ times}
        % }
        % \]
        \includegraphics[width=\linewidth]{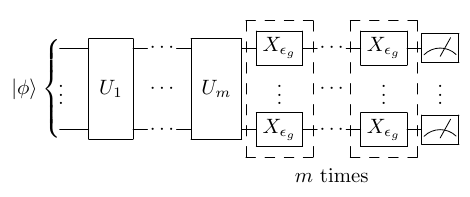}
        \caption{\footnotesize The converted model under commutativity condition\label{fig:bfm}}
     \end{subfigure}
        \caption{Circuit illustration of a bit-flip noise model and its restricted equivalence under commutativity assumption.}
        \label{fig:bfms}
\end{figure}

The expected distribution function $\Tilde{p}$ of circuit $U_m \cdots U_1\ket{\phi}$ with up to $m$ layers of bit-flip errors can be recursively defined by
\begin{align*}
    \Tilde{p}^{(1)}(x) &:= E_x[p(y^{(1)})] \\
    \Tilde{p}^{(j)}(x) &:= E_x[p^{(j-1)}(y^{(j-1)})] &&\text{for } j = 2,\dots,m \\
    \Tilde{p}(x) &= \Tilde{p}^{(m)}(x) =  E_x[p^{(m-1)}(y^{(m-1)})]
\end{align*}
where $p$ is the error-free output distribution,
\begin{displaymath} 
    y^{(1)}_i = \begin{cases}
    x_i &\text{ with probability $1 - \eg$} \\
    \neg x_i &\text{ with probability $\eg$}
    \end{cases}, \text{ and }
    y^{(j)}_i = \begin{cases}
    y^{(j-1)}_i &\text{ with probability $1 - \eg$} \\
    \neg y^{(j-1)}_i &\text{ with probability $\eg$}
    \end{cases},
\end{displaymath}
for $i = 1,...,n$ and $j = 2,...,m$ (to avoid the confusion, the superscripts on $\Tilde{p}$ are just indices). Because the expectations are all over $x$
not $s$, we can repeat the process in Section~\ref{sec:FE} $m$ times. 
Each repetition provides a $(1-2\epsilon_{g})^{|s|}$ term in the multiplication:
\begin{equation}
    \Tilde{p}(x) =  \sum_{s \in \{0,1\}^n} \left[ \left(\prod_{j = 1}^{m}(1-2\epsilon_{g})^{|s|}\right) \hat{p}(s)(-1)^{s.x} \right]  = \sum_{s \in \{0,1\}^n}  (1-2\epsilon_{g})^{|s| \cdot m}\hat{p}(s)(-1)^{s.x}. \label{eq:pm}
\end{equation}
Eq.~\eqref{eq:pm} is also straightforward to compatible with the case when each layer of bit-flip errors have a different error rate by indexing $\eg$ with $j$
\begin{displaymath}
    \Tilde{p}(x) =  \sum_{s \in \{0,1\}^n} \left[ \left(\prod_{j = 1}^{m}(1-2\epsilon_{g,j})^{|s|}\right) \hat{p}(s)(-1)^{s.x} \right].
\end{displaymath}

\subsubsection{Bit-flip Error Filter}

Let $j_b$ be the binary representation of a
non-negative integer $j$. Given $\epsilon_g$ and $\ptilde(x)$ for all $x \in
\{0,1\}^n$, it is possible to recover the noise-free outcomes of a QC algorithm. The
first step is to solve for $\phat(s)$.  With known $\epsilon_g$, $\ptilde(x)$,
and $x$, a linear system derived from Eq.~\eqref{eq:pm} can be built as following:
\begin{equation}
    G\rhohat = \rhotilde \label{eq:denoise-gate}
\end{equation}
where
\begin{align*}
%   & G_{ij} := (1 - \eg)^{(j-1)m} (-1)^{(i-1)_b.(j-1)_b}\quad \text{for $i \in \{1,...,2^n\}$ and $j \in \{1,...,2^n\}$} \\
    & G_{ij} := (1 - 2\eg)^{|(j-1)_b|m} (-1)^{(i-1)_b.(j-1)_b}\quad \text{for $i \in \{1,...,2^n\}$ and $j \in \{1,...,2^n\}$} \\
    % G &:= \begin{bmatrix}
    % \left[(1-\epsilon_{g})^{|0...0|}\right]^m(-1)^{(0...0).(0...00)}  & \cdots & \left[(1-\epsilon_{g})^{|1...1|}\right]^m(-1)^{(1...1).(0...00)}\\
    % \left[(1-\epsilon_{g})^{|0...0|}\right]^m(-1)^{(0...0).(0...01)}  & \cdots & \left[(1-\epsilon_{g})^{|1...1|}\right]^m(-1)^{(1...1).(0...01)}\\
    % \vdots  &\ddots &\vdots \\
    % \left[(1-\epsilon_{g})^{|0...0|}\right]^m(-1)^{(0...0).(1...11)}  & \cdots & \left[(1-\epsilon_{g})^{|1...1|}\right]^m(-1)^{(1...1).(1...11)}\\
    % \end{bmatrix} \\
  &   \rhohat := \begin{pmatrix}
    \phat(0...00) \\  \phat(0...01) \\  \vdots \\ \phat(1...11)
    \end{pmatrix} , \qquad
    \rhotilde := \begin{pmatrix}
    \ptilde(0...00) \\ \ptilde(0...01) \\ \vdots \\ \ptilde(1...11)
    \end{pmatrix},
\end{align*}
$G \in [-1,1]^{2^n \times 2^n}$, $\rhohat \in
\left[-\frac{1}{2^n},\frac{1}{2^n}\right]^{2^n}$, and $\rhotilde \in
[0,1]^{2^n}$. Using the algorithm to be introduced in Section~\ref{sec:est}, we can
estimate the value of $\eg$ to construct matrix $G$. 
Using a sufficient number of measurements, we can compute vector $\rhotilde$. 
Thus, by solving Eq.~\eqref{eq:denoise-gate} and substituting the result into Eq.~\eqref{eq:pFE}, we can then re-construct the noise-free distribution function $p(x)$ for all $x \in \{0,1\}^n$. 
The following lemma implies that the solution of Eq.~\eqref{eq:denoise-gate} always exists.

\begin{lemma} \label{lemma:Ginv}
$G$ is full-rank for all $n \geq 1$.
\end{lemma}
\begin{proof}
We decompose $G$ as
\begin{displaymath}
    G = G_1^{(n)} \circ G_2^{(n)},
\end{displaymath}
where $\circ$ is element-wise multiplication, $G_1^{(n)} \in [0,1]^{2^n \times 2^n}$, $G_2^{(n)} \in \{-1,1\}^{2^n \times 2^n}$, and
\begin{align*}
    % (G^{(n)}_1)_{ij} &:= (1 - \eg)^{(j-1)m} \quad\text{for $i \in \{1,...,2^n\}$ and $j \in \{1,...,2^n\}$} \\
    (G^{(n)}_1)_{ij} &:= (1 - 2\eg)^{|(j-1)|_b m} \quad\text{for $i \in \{1,...,2^n\}$ and $j \in \{1,...,2^n\}$} \\
    (G^{(n)}_2)_{ij} &:= (-1)^{(i-1)_b(j-1)_b} \quad\text{for $i \in \{1,...,2^n\}$ and $j \in \{1,...,2^n\}$}
\end{align*}
We start with $G_2^{(n)}$ when $n = 1$. It is easy to examine that
\begin{displaymath}
    G^{(1)}_2 = \begin{bmatrix} 1 &1 \\ 1&-1 \end{bmatrix}
\end{displaymath}
is full-rank. Recall that each entry of $G^{(n)}_2$ is $(-1)^{s.x}$ for binary
numbers $s$ and $x$, and each row of $G^{(n)}_2$ shares the same $x$ while each column shares the same $s$. 
Following the little-endian convention, the 16 entries in $G^{(2)}_2$ can be divided into four divisions equally based on their position
\begin{displaymath}
    \left[\begin{array}{c|c}
    s_2 = 0, x_2 = 0 & s_2 = 0, x_2 = 1 \\ \hline
    s_2 = 1, x_2 = 0 & s_2 = 1, x_2 = 1
    \end{array}\right].
\end{displaymath}
Namely,
\begin{displaymath}
    G^{(2)}_2 = \left[\begin{array}{c|c}
    (-1)^{0 \cdot 0}G^{(1)}_2 & (-1)^{0 \cdot 1}G^{(1)}_2 \\ \hline
    (-1)^{1 \cdot 0}G^{(1)}_2 & (-1)^{1 \cdot 1}G^{(1)}_2
    \end{array}\right] = \left[\begin{array}{c|c}
    G^{(1)}_2 & G^{(1)}_2 \\ \hline
    G^{(1)}_2 & -G^{(1)}_2
    \end{array}\right].
\end{displaymath}
Similarly, we can have
\begin{equation}
    G^{(n)}_2 = \left[\begin{array}{c|c}
    G^{(n-1)}_2 & G^{(n-1)}_2 \\ \hline
    G^{(n-1)}_2 & -G^{(n-1)}_2
    \end{array}\right] \label{eq:Gnrelation}
\end{equation}
Since $G^{(n)}_2 \in \{-1,1\}^{2^n \times 2^n}$, if $G^{(n-1)}_2$ is full-rank, the structure in Eq.~\eqref{eq:Gnrelation} implies $G^{(n)}_2$ is full-rank. As $G^{(1)}_2$ is also full-rank, by induction, $G^{(n)}_2$ is full-rank for all $n \geq 1$.

% $\eg \in (0,1)$, $1-\eg \neq 0$
%Note that for an arbitrary column $j$ of $G$, it is equivalent to column $j$ of $G_2^{(n)}$ multiplied by $(1 - \eg)^{(j-1)m}$. 
% Note that the $j$th column of $G$, it the $j$th column of $G_2^{(n)}$ multiplied by $(1 - \eg)^{(j-1)m}$. 
Note that the $j$th column of $G$ is the $j$th column of $G_2^{(n)}$ multiplied by $(1 - 2\eg)^{|(j-1)|_b m}$ and $1 - 2\eg \neq 0$. 
As all columns of $G_2^{(n)}$ are linearly independent, their non-zero multiples are linearly independent, too. 
Namely, all columns of $G$ are linearly independent, so $G$ is full-rank for all $n \geq 1$.
\end{proof}

%With Lemma~\ref{lemma:Ginv}, \eqref{eq:denoise-gate} always has a solution. 
However, similar to the problem in Section~\ref{sec:me}, solving Eq.~\eqref{eq:denoise-gate} cannot guarantee a meaningful $\rhotilde$, that is, a $\rhotilde \in [0,1]^{2^n}$. Nevertheless, we can consider a optimization problem instead.
\begin{equation}
\begin{aligned} \label{eq:gate-opt}
%    \rhohat^* := \argmin &\ \|G^{-1}\rhotilde - \rhohat\|_2 \\ \text{s.t. }&\\
    \rhohat^* := \argmin &\ \|G\rhohat - \rhotilde \|_2 \\ \text{s.t. }&\\
    \sum_{i = 1}^{2^n}\sum_{j = 1}^{2^n}&\rhohat_i(-1)^{(i-1)_b.(j-1)_b} = 1 \\
    \sum_{i = 1}^{2^n}&\rhohat_i(-1)^{(i-1)_b.(j-1)_b} \geq 0 \quad \text{for all $j \in \{1,...,2^n\}$} 
\end{aligned}
\end{equation}
As an example, when $n = 1$, Eq.~\eqref{eq:pFE} yields
\begin{equation}\label{eq:gate-n1-p0}
\begin{aligned}
    p(0) &= \phat(0) + \phat(1) \\
    p(1) &= \phat(0) - \phat(1),
\end{aligned}    
\end{equation}
so $\phat(0)$ is always $\frac{1}{2}$ as $p(0) + p(1) = 1$. Thus, when $n=1$, Eq.~\eqref{eq:gate-opt} can be simplified as
\begin{equation}
\begin{aligned} \label{eq:gate-optn1}
%    \rhohat^* := \argmin_{\rhohat_1 = \frac{1}{2}, -\frac{1}{2}  \leq \rhohat_2 \leq \frac{1}{2}} \ \|G^{-1}\rhotilde - \rhohat\|_2 
    \rhohat^* := \argmin_{\rhohat_1 = \frac{1}{2}, -\frac{1}{2}  \leq \rhohat_2 \leq \frac{1}{2}} \ \|G\rhohat-\rhotilde \|_2 
\end{aligned}
\end{equation}

%===============================================================================

\section{Estimating Distributions of Noise Parameters}\label{sec:est}

The bit-flip gate error model Eq.~\eqref{eq:pm} and the measurement error model
Eq.~\eqref{eq:me} together provide us forward models to propagate noise in QC
algorithms.
%predict noisy output of circuits like one in Figure~\ref{fig:bfmp}. 
Based on these forward models and measurement results from a QC device, we can
filtering out measurement errors and, in some scenarios, bit-flip gate errors,
as such to recover noise-free information.
%from a noisy output distribution of some circuits. 
Here, a critical step is to identify model parameters $\epsilon_g, \epsilon_{m0}$
and $\epsilon_{m1}$
%in Eq.~\eqref{eq:pm} and Eq.~\eqref{eq:me} 
using repeated measurements of a testing circuit. 
The Bayesian approach is suited to solving this inverse problem. 
In this work, we will use both the standard Bayesian inference and a novel
Bayesian approach called \emph{consistent Bayesian}~\cite{Butler2018} to infer
these parameters. 

\subsection{Computational Framework}

The Bayesian inference considers model parameters conditioned on data $d$ as the
posterior distribution $\pi(\lambda|d)$, which is proportional to the product
of the prior distribution parameters $\pi(\lambda)$ and the likelihood
$\pi(d|\lambda)$, i.e. $\pi(\lambda|d)\propto \pi(\lambda)\pi(d|\lambda)$. 
It infers the posterior distribution using the stochastic map 
$d = Q(\lambda) + \varepsilon$, where $Q$ is the quantity of interest (QoI) 
and $\varepsilon$ is an assumed error model. In our case, $\lambda$ represents
model parameters $\epsilon_g, \epsilon_{m0}$ and $\epsilon_{m1}$, $d$ is the
measured data collected from the device, and $\pi(d|\lambda)$ characterize the
difference between forward model output and the data. 

Unlike the standard Bayesian inference, the consistent Bayesian directly 
inverts the observed stochasticity of the data, described as a probability
measure or density, using the deterministic map $Q(\lambda)$. This
approach also begins with a prior distribution, denoted
as $\pi_{\bm\Lambda}^{\text{prior}}(\lambda)$, on the model parameters, which is
then updated to construct a posterior distribution
$\pi_{\bm\Lambda}^{\text{post}}(\lambda)$.
But its posterior distribution takes a different form:
\begin{equation}
  \label{eq:con_bay}
  \pi_{\bm\Lambda}^{\text{post}}(\lambda) =
 \pi_{\bm\Lambda}^{\text{prior}}(\lambda)\dfrac{\pi_{\mathcal{D}}^{\text{obs}}(Q(\lambda))}{\pi_{\mathcal{D}}^{Q(\text{prior})}(Q(\lambda))},
\end{equation}
where $\lambda\in\bm\Lambda$ and $\mathcal{D}$ is the space of the observed
data. Each terms in Eq.~\eqref{eq:con_bay} are explained as follows:
\begin{itemize}[itemsep=0pt, topsep=1pt]
  \item $\pi_{\mathcal{D}}^{Q(\text{prior})}$ denotes 
the push-forward of the prior through the model and represents a forward
    propagation of uncertainty. It represents how the prior knowledge of likelihoods of parameter values defines a likelihood of model outputs.
    % in the statistical sense, and in the Bayesian sense can be interpreted as the relative evidence, by treating $\lambda$ as the predictor variables and $Q(\lambda)$ as the response variables.
  \item $\pi_{\mathcal{D}}^{\text{obs}}$ is the observed probability
    density of the QoI. It describes the likelihood that
    the output of the model corresponds to the observed data.
\end{itemize}

%The key difference of these two approach is that even though both approaches result in a posterior distribution of the desired parameters, the former consider the exact parameters as constants while the latter treat these parameters as random variables intrinsically. {\color{blue}Rewrite the key difference and the following sentence} In the error model for QC, these two treatments represent different understanding of the underlying physical system (i.e., the device).  Namely, the standard Bayesian assumes the exact error rates are deterministic, while the other one assume them to be random variables that may vary.

\subsection{Implementation Details}
We take a noisy one-qubit gate $\Utilde$ (its noise-free version is denoted by $U$) as an example. 
%The probability of the occurrence of a depolarizing error is $\epsilon_g$ as shown in Figure~\ref{fig:testing-circ}. 
Suppose we use this gate to build a testing circuit as shown in Figure~\ref{fig:testing-circ}.
%We adopt a Bayesian approach from \cite{Butler2018} to infer the (posterior) distributions of noise parameters. 
We set the QoI in our case to be the probability of measuring $0$ from the
testing circuit.
% We infer noise parameters as they are random variables according to some distributions instead of, like standard Bayesian method, some unknown fixed values. Our discussion and comparison regarding this issue is in Section~\ref{sec:exp}.
Assume that the measurement operator in the testing circuit is associated with measurement errors $\ez$ and $\eo$. 
%when measures $\ket{0}$ and $\ket{1}$ respectively. 
Let $\lambda := (\epsilon_g,\ez,\eo)$ be the tuple of noise parameters that we want to infer. 
Note that if $\Utilde$ is a gate like Hadamard gate, the bit-flip gate error in
theory will not affect the measurement outcome for testing circuit, which means
we only need to infer $\ez$ and $\eo$ in this case. %but not $\eg$ in this case. 
In terms of measurement error rates estimation, we provide a choice of testing
circuit in Section~\ref{sec:meas-exp} consisting of a single testing circuit for
$n$ qubits, which dramatically reduces the number of testing circuits compared
with the fully correlated setting.
Let $\Lambda:=(0,1) \times (0,1) \times (0,1)$ denote the space of noise parameters and $\mathcal{D} := [0,1]$ denote the space of QoI.
Finally, we use $Q: \Lambda \rightarrow \mathcal{D}$ to denote a general function
combining Eq.~\eqref{eq:pm} and Eq.~\eqref{eq:me} that compute the probability
of measuring $\ket{0}$ when testing circuit has bit-flip gate error and measurement error.

\begin{figure}
    \centering
    % \[
    % \Qcircuit @C=1em @R=.7em {
    % &\lstick{\ket{0}} &\gate{\Utilde}   &\meter  
    % }
    % \]
    \includegraphics[width=0.2\linewidth]{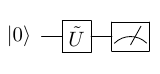}
    \caption{A testing circuit example}
    \label{fig:testing-circ}
\end{figure}

The overall algorithm consists of two parts. In the first part, $L$ number of
QoI's, denoted by $q_j$ ($j = 1,...,L$) are generated from $L$ number of prior
$\lambda$'s, denote as $\lambda_j$ for $j = 1,...,L$, with function $Q$.
Then, the distribution $\pi_{\mathcal{D}}^{Q(\text{prior})}$ is estimated by
Gaussian kernel density (KDE) using $q_j$. %~\cite[Algorithm 1]{Butler2018}. 
Next, in the second part, prior $\lambda_j$'s are either rejected or accepted
based on Eq.~\eqref{eq:con_bay},
%certain ratio derived from the previous part, 
and those accepted prior $\lambda_j$'s are the posterior noise parameters that
we are looking for. The distribution
$\pi^{\text{obs}}_{\mathcal{D}}$ is the observed probability of measuring $\ket{0}$, i.e.,
Gaussian KDE of data. The algorithm is summarized in 
Algorithm~\ref{alg:noise_params}, which is an implementation of Algorithms 1 and
2 in~\cite{Butler2018}.
\begin{algorithm}
  \caption{Consistent Bayesian inference for error model parameters.}
    %~\cite[Algorithm 1\&2]{Butler2018}}
    \label{alg:noise_params}
    \begin{algorithmic}
    %  \STATE Given a set of prior $\lambda_j: \ j = 1,...,L$, Gaussian KDE $\pi^{\text{obs}}_{\Dcal}$ from data, and input state of testing circuit;
      \STATE Given a set of prior $\lambda_j$ ($j = 1,...,L$), Gaussian KDE
      $\pi^{\text{obs}}_{\mathcal{D}}$ of the observed QoI (i.e., data), model 
      function $Q$ (i.e., combination of Eq.~\eqref{eq:pm}, 
      Eq.~\eqref{eq:me}, testing circuit and its input state);
     \FOR{$j = 1$ to $L$}
%     \STATE Use \eqref{eq:pm} and \eqref{eq:me} with $\lambda_j$ and input state of testing circuit to compute $q_j$;
     \STATE Use $Q(\lambda_j)$ to compute $q_j$;
     \ENDFOR
     \STATE Generate Gaussian KDE $\pi_{\mathcal{D}}^{Q(\text{prior})}$ from $q_j$'s;
     \STATE Estimate $\mu:= \max_{\lambda \in \Lambda} \frac{\pi^{\text{obs}}_{\mathcal{D}}(Q(\lambda))}{\pi_{\mathcal{D}}^{Q(\text{prior})}(Q(\lambda))}$;
     \FOR{$k = 1$ to $L$}
     \STATE Generate a random number $\zeta_k \in [0,1]$ from a uniform distribution;
     \STATE Compute ratio $\eta_k := \frac{1}{\mu}\cdot \frac{\pi^{\text{obs}}_{\mathcal{D}}(Q(\lambda_k))}{\pi_{\mathcal{D}}^{Q(\text{prior})}(Q(\lambda_k))}$;
     \IF{$\eta_k > \zeta_k$}
     \STATE Accept $\lambda_k$;
     \ELSE
     \STATE Reject $\lambda_k$;
     \ENDIF
     \ENDFOR
     \STATE\textbf{output} Accepted noise parameter $\lambda_k$'s.
    \end{algorithmic}
\end{algorithm}

In practice, the prior $\lambda_j$ are randomly generated from some relatively flat normal distributions due to the little knowledge of its actual characterization. Thus, for Qubit $i$, suppose we have estimated gate and measurement error rates $(\egi^0, \ezi^0, \eoi^0)$ from past experience and their variances $(\sigma_{\egi}, \sigma_{\ezi}, \sigma_{\eoi})$ that make curves flat, the prior distributions are
\begin{align*}
     \ezi &\sim N(\ezi^0, \sigma_{\ezi}^2), \\
     \eoi &\sim N(\eoi^0, \sigma_{\eoi}^2), \\
     \egi &\sim N(\egi^0, \sigma_{\egi}^2). 
\end{align*}
In this setting, the acceptance rates of all experiments in Section~\ref{sec:exp} range from $10\%$ to $35\%$. This is high enough to select sufficient number of posterior parameters in this study.

To demonstrate and compare the difference between the results of consistent and
standard Bayesian algorithms, we also use the same priors and observation
datasets to infer noise parameters via the standard Bayesian. 
%The function $Q$ is linear due to the linearity of Eq.~\eqref{eq:me} and Eq.~\eqref{eq:pm}, so the standard Bayesian inference on this problem is equivalent to the linear regression. 
For a single Qubit $i$, 
%the predictor of the linear regression is the ideal probability of measuring $0$, which is trivially determined by the fixed testing circuit and its input state, and response is the same as the observed QoI. 
let $(x_j, y_j)$ for $j = 1,...,J$ represent $J$ number of data pairs, where
$x_j$ is the theoretical probability of measuring $\ket{0}$ and $y_j$ is the
observed probability of measuring $\ket{0}$. As discussed in 
Eq.~\eqref{eq:error_mat0} and Eq.~\eqref{eq:gate-n1-p0}, we have 
\begin{equation}
    % y_j = ((0.5 + (1-\egi)^{m} (x_j-0.5))  (1-\ezi) + (0.5 - (1-\egi)^{m} (x_j - 0.5)) \eoi + \varepsilon_j \label{eq:yjxj}
    y_j = ((0.5 + (1-2\egi)^{m} (x_j-0.5))  (1-\ezi) + (0.5 - (1-2\egi)^{m} (x_j - 0.5)) \eoi + \varepsilon_j \label{eq:yjxj}
\end{equation}
where $m$ is the number of repetitions of the gate in the testing circuit ($m =
1$ in Figure~\ref{fig:testing-circ}) and $\varepsilon_j\sim
N(0,\sigma_{\varepsilon}^2)$ represents noise in general with standard deviation
$\sigma_{\varepsilon} \geq 0$. We use $\text{Cauchy}(0,1)$ as the prior
distribution of $\sigma_{\varepsilon}$. Eq.~\eqref{eq:yjxj} yields the following
likelihood function
\begin{displaymath}
    f(\vec{y}|\vec{x}, \ezi, \eoi, \egi, \sigma_{\varepsilon}) = \prod_{j =
    1}^{J}f_j(y_j|x_j, \ezi, \eoi, \egi, \sigma_{\varepsilon}),
\end{displaymath}
where each $f_j$ is the probability density function (PDF)
\begin{displaymath}
    % N(((0.5 + (1-\egi)^{m} (x_j-0.5))  (1-\ezi) + (0.5 - (1-\egi)^{m} (x_j - 0.5)) \eoi, \sigma_{\varepsilon}^2).
    N(((0.5 + (1-2\egi)^{m} (x_j-0.5))  (1-\ezi) + (0.5 - (1-2\egi)^{m} (x_j - 0.5)) \eoi, \sigma_{\varepsilon}^2).
\end{displaymath}
In this work, we use \texttt{RStan} package in \texttt{R}~\cite{RStan, R} to
implement the standard Bayesian inference.
%for each experiment in Section~\ref{sec:exp}. The \texttt{RStan} parameters and algorithm for standard Bayesian are 
which is summarized in Algorithm~\ref{alg:std-Bayesian}.

\begin{algorithm}
    \caption{Standard Bayesian inference for error model parameters.}
    \label{alg:std-Bayesian}
    \begin{algorithmic}
     \STATE \textbf{Data.} 
     \STATE Number of repetitions of gates in the testing circuit $m$, %Number of data points $J$, 
      theoretical probabilities of measuring $\ket{0}$ $x_j$, observed probabilities
      of measuring $\ket{0}$ $y_j (j = 1,..., J$), prior mean of noise parameters $(\ezi^0,\eoi^0, \egi^0)$, and prior variance of noise parameters $(\sigma_{\egi}, \sigma_{\ezi}, \sigma_{\eoi})$.
     \STATE \textbf{Model Parameters.} 
     \STATE posterior $(\ezi,\eoi,\egi) \in (0,1)^3$ and $\sigma_{\varepsilon} \geq 0$.
     \STATE \textbf{Prior Distributions.}
     \STATE $\sigma_{\varepsilon} \sim \text{Cauthy}(0,1)$; %\# mean 0
     \STATE $\ezi \sim N(\ezi^0, \sigma_{\ezi}^2)$; %\# same prior as in consistent Bayesian
     \STATE $\eoi \sim N(\eoi^0, \sigma_{\eoi}^2)$; %\# same prior as in consistent Bayesian
     \STATE $\egi \sim N(\egi^0, \sigma_{\egi}^2)$; %\# same prior as in consistent Bayesian
     \STATE \textbf{Likelihood Function.}
     \STATE Only measurement errors:
     \STATE $\forall j$, $y_j \sim N( x_j(1-\ez) + (1-x_j) \eo, \sigma_{\varepsilon}^2)$.
     \STATE Gate and measurement errors:
    %  \STATE $\forall j$, $y_j \sim N(((0.5 + (1-\eg)^{m} (x_j-0.5))  (1-\ez) + (0.5 - (1-\eg)^{m} (x_j - 0.5)) \eo, \sigma_{\varepsilon}^2)$.
    \STATE $\forall j$, $y_j \sim N(((0.5 + (1-2\eg)^{m} (x_j-0.5))  (1-\ez) + (0.5 - (1-2\eg)^{m} (x_j - 0.5)) \eo, \sigma_{\varepsilon}^2)$.
     \STATE \textbf{Stan parameters.} 
     \STATE Default No-U-Turn Sampler, 10,000 iterations, 2000 warm-up iterations, \texttt{adapt\_delta $= 0.99$}, and other parameters are default.
    \end{algorithmic}
\end{algorithm}

%Need to add a description of the standard Bayesian inference

%===============================================================================

\section{Experiments}\label{sec:exp}

Because using our bit-flip error model only is not sufficient for the analyzing
gate errors in a complicate algorithm like Grover's search or QAOA, the
inference for gate errors is performed for a few prototype circuits. For more
sophisticated algorithms, we only investigate the measurement error. All
experiments are conducted on IBM's 5-qubit quantum computer \texttt{ibmqx2}. We
compare both the consistent Bayesian (Algorithm~\ref{alg:noise_params}) and the
standard Bayesian method (Algorithm~\ref{alg:std-Bayesian}) with the measurement
error filter in Qiskit \texttt{CompleteMeasFitter}~\cite{Qiskit} and the method
in~\cite{QDT2020} based on quantum detector tomography (QDT) to demonstrate the
efficiency of our approaches.

\subsection{Measurement Errors Filtering Experiment}\label{sec:meas-exp}

\subsubsection{Construction of Error Filter}
\begin{figure}
    \centering
    % \[
    % \Qcircuit @C=1em @R=1em {
    % &\gate{H}   &\meter \\
    % &\vdots     &\vdots \\
    % &\gate{H}   &\meter \inputgroupv{1}{3}{.75em}{1.5em}{\ket{0}^{\otimes n}}
    % }
    % \]
    \includegraphics[width=0.2\linewidth]{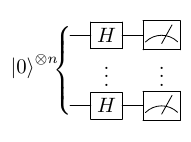}
    \caption{Testing circuit for measurement error parameter inference\label{fig:testing-circ-H}}
    \label{fig:test1}
\end{figure}
We use the circuit in Figure~\ref{fig:testing-circ-H} to infer measurement error
parameters in every single qubit, i.e., $\ezi,\eoi$ for $i \in \{1,2,3,4\}$ on \texttt{ibmqx2}.
Here, $H$ is the Hadamard gate for each qubit. Theoretically, the observed
results of $H\ket{0}$ is invariant under bit-flip and phase-flip errors.
Consequently, in this case, only the measurement error 
affects the distribution of measurement outputs, and we do not infer
gate error rate $\eg$. The testing circuit is executed for $1024 \times 128$
times, where the fraction of measuring $0$ in each ensemble consisting of $1024$
runs provides estimated probability of measuring $0$ from the testing circuit.
Thus, we have $128$ data points in total, i.e., $L=128$ in
Algorithm~\ref{alg:noise_params} or $J=128$ in Algorithm~\ref{alg:std-Bayesian}.
For qubit $i$, the prior $(\ezi,\eoi) \subseteq (0,1) \times (0,1)$ are
random number from truncated normal distribution $N(\ezi^0,0.1^2)$ and
$N(\eoi^0,0.1^2)$, respectively, where $\ezi^0$ and $\eoi^0$ are corresponding
values provided by IBM in Qiskit API \texttt{IBMQbackend.properties()} after the
daily calibration. Then, we use Algorithm~\ref{alg:noise_params} to generate the
posterior distributions. We note that in this test, the results by the
consistent Bayesian is very close to the standard Bayesian, so we present the
former only.

Figure~\ref{fig:p0p1} displays the joint and marginal distribution of posterior
distributions of error model parameters for qubits 1-4 using the consistent
Bayesian approach. %two tuples of these posteriors:
%\begin{enumerate}[topsep=2pt, itemsep=0pt]
%    \item posterior mean: the tuple of posterior $(\ezi,\eoi)$ that are closest to the mean of marginal distributions of posterior $\ezi$ and $\eoi$ in Euclidean distance;
%    \item posterior Maximum A Posteriori (MAP): the tuple of posterior $(\ezi,\eoi)$ that are closest to the peaks of marginal distributions of posterior $\ezi$ and $\eoi$ in Euclidean distance;
%\end{enumerate}
%The values of posterior mean and posterior MAP are reported in 
Using these posterior distributions of error model parameters, we
can compute the posterior distribution of the QoI by substituting samples of
these distributions in the forward model $Q$.
Figure~\ref{fig:QoI} shows that these posteriors of the QoI, denoted as
$\pi_{\mathcal{D}}^{Q(\text{post})}$, approximate the distribution of the
observed data $\pi_{\mathcal{D}}^{\text{obs}}$ very well.

\begin{figure}
     \centering
     \begin{subfigure}[b]{0.45\textwidth}
         \centering
         \includegraphics[width=\textwidth]{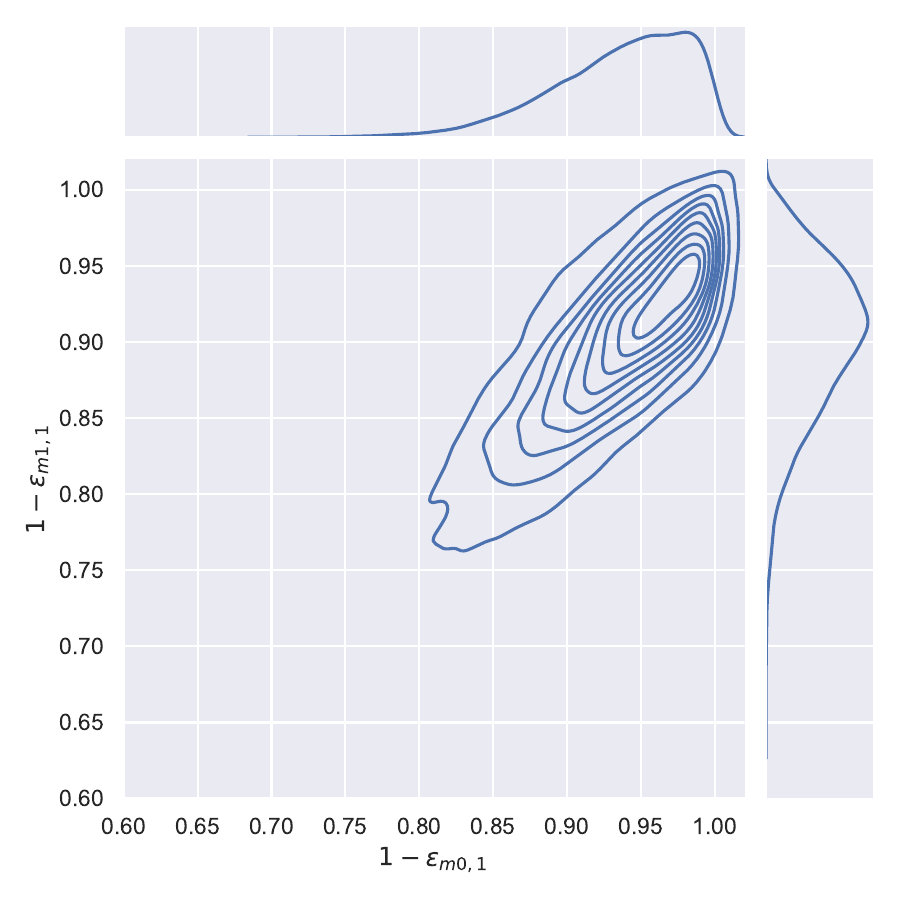}
         \caption{Qubit 1}
     \end{subfigure} \quad
     \begin{subfigure}[b]{0.45\textwidth}
         \centering
         \includegraphics[width=\textwidth]{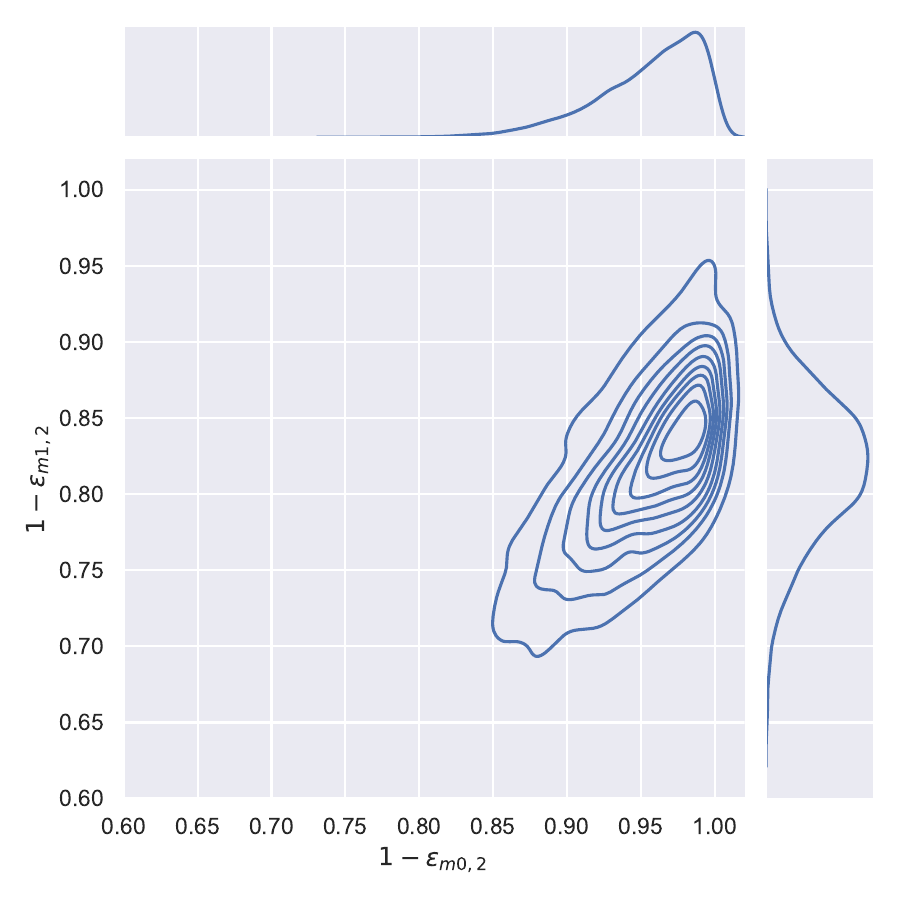}
         \caption{Qubit 2}
     \end{subfigure}
     \\
     \begin{subfigure}[b]{0.45\textwidth}
         \centering
         \includegraphics[width=\textwidth]{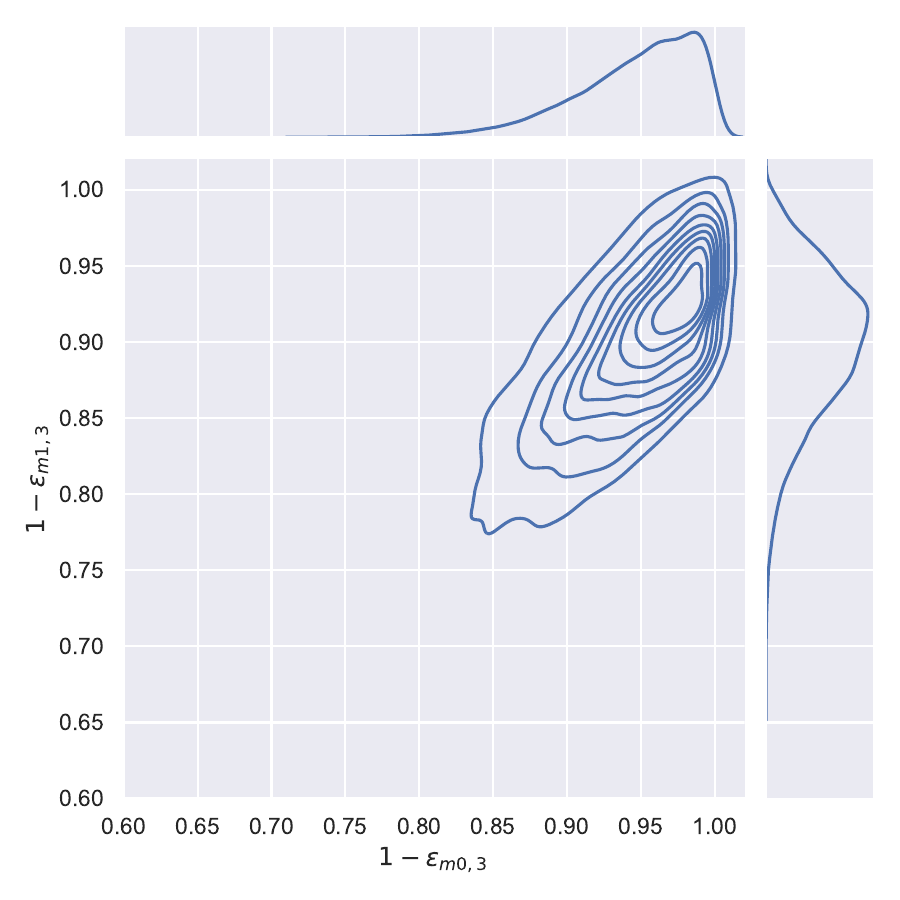}
         \caption{Qubit 3}
     \end{subfigure}\quad
     \begin{subfigure}[b]{0.45\textwidth}
         \centering
         \includegraphics[width=\textwidth]{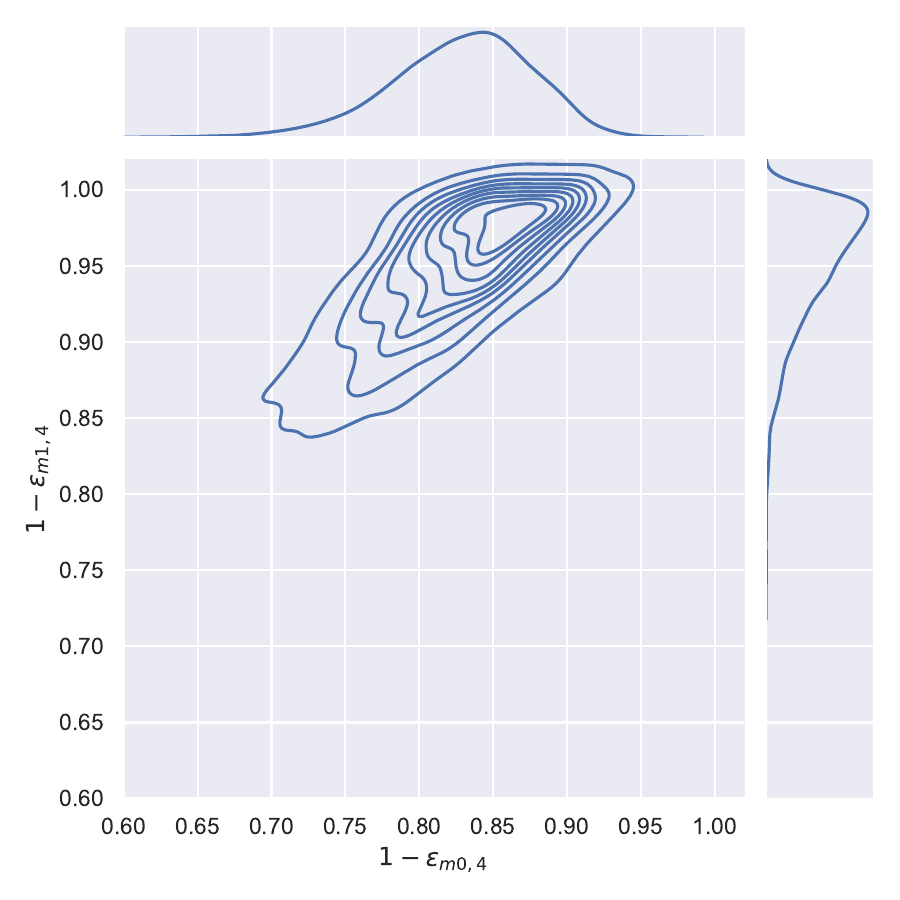}
         \caption{Qubit 4}
     \end{subfigure}
      %  \caption{Distribution of posteriors}
        \caption{Joint and marginal posterior distributions of measurement error parameters  in the testing circuit shown in Figure~\ref{fig:test1}. Here,
        $(1-\ezi, 1-\eoi)$ are shown for demonstration purpose.}
        \label{fig:p0p1}
\end{figure}

\begin{figure}
     \centering
     \begin{subfigure}[b]{0.49\textwidth}
         \centering
         \includegraphics[width=\textwidth]{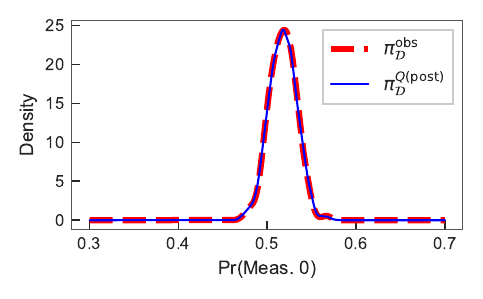}
         \caption{Qubit 1}
     \end{subfigure}
     \begin{subfigure}[b]{0.49\textwidth}
         \centering
         \includegraphics[width=\textwidth]{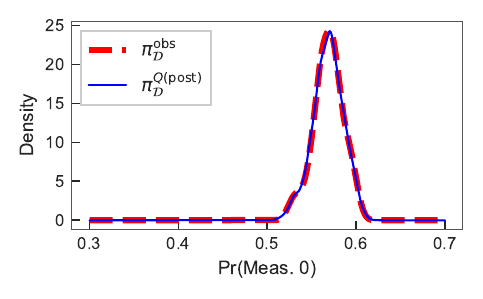}
         \caption{Qubit 2}
     \end{subfigure}
     \\
     \begin{subfigure}[b]{0.49\textwidth}
         \centering
         \includegraphics[width=\textwidth]{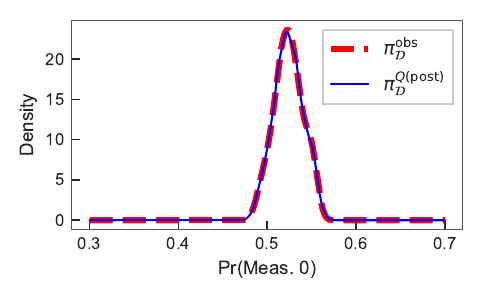}
         \caption{Qubit 3}
     \end{subfigure}
     \begin{subfigure}[b]{0.49\textwidth}
         \centering
         \includegraphics[width=\textwidth]{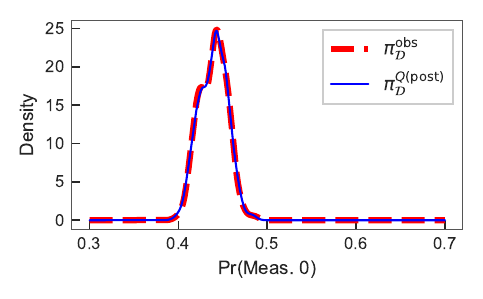}
         \caption{Qubit 4}
     \end{subfigure}
        \caption{%Gaussian KDE of 
          PDF of the QoI (i.e., the probability of measuring $\ket{0}$) obtained from data
          ($\pi_{\mathcal{D}}^{\text{obs}}$)
          and from evaluating $Q$ using inferred measurement error parameters
          ($\pi_{\mathcal{D}}^{Q(\text{post})}$).}
        \label{fig:QoI}
\end{figure}

For a more quantitative comparison, we list the posterior mean and the maximum 
\emph{a posteriori} probability (MAP) in
Table~\ref{tab:QoI}. We can see, in general, $\eoi$ is higher than $\ezi$, which
is consistent with the description in~\cite{Hicks2021}.
Also, Tables~\ref{tab:QoI} presents the Kullback–Leibler (KL) divergence 
between PDFs of the observed data $\pi_{\mathcal{D}}^{\text{obs}}$ and the
posterior distribution of the QoI $\pi_{\mathcal{D}}^{Q(\text{post})}$ 
for each qubit in Figure~\ref{fig:QoI}, which illustrates the accuracy of 
our error model.
\begin{table}
    \centering
    \caption{Outcomes from the consistent Bayesian inference}
    \label{tab:QoI}
    \begin{tabular}{lcccc} \toprule
         &Qubit 1 &Qubit 2 &Qubit 3 &Qubit 4\\ \midrule
    KL-div($\pi_\mathcal{D}^{Q(\text{posterior})},\pi_\mathcal{D}^{\text{obs}}$) &$0.001014$ &$0.002243$ &$0.000777$ &$0.001610$ \\ 
    Post. Mean $(1-\ezi,1-\eoi)$ &(0.9354,0.9009) &(0.9537,0.8184) &(0.9457,0.8976) &(0.8272,0.9492) \\
    Post. MAP $(1-\ezi,1-\eoi)$ &(0.9797,0.9128) &(0.9863,0.8243) &(0.9858,0.9180) &(0.8426,0.9846) \\ \bottomrule
    \end{tabular}
\end{table}

%QoI, i.e., the probability of measuring $0$, quite well as they almost coincide with the data collected from the device. %We can see our posterior basically captures every perspective of data. 

In this test, our prior distribution $N(\ezi^0,0.1^2)$ and $N(\eoi^0,0.1^2)$ 
are quite flat and not informative. This is because the vendor-provided 
$\ezi^0$ and $\eoi^0$ are not always good estimations. This can be verified by
the error mitigation results.
When we use relation Eq.~\eqref{eq:me} and Eq.~\eqref{eq:me-filter} to construct
measurement error filters using the vendor-provided $(\ezi^0,\eoi^0)$ and our
posteriors, then apply those filters on the $128$ outputs of circuit in
Figure~\ref{fig:testing-circ-H} (i.e., $128$ observed probability of
measuring $\ket{0}$), we obtain different results as shown in
Figure~\ref{fig:DQoI}. The theoretical probability of measuring $0$ for circuit
in Figure~\ref{fig:testing-circ-H} is $0.5$, but the provided parameters
rarely gives this value, and its mean and peak of Gaussian KDE are not even
close to $0.5$. On the other hand, the filters created by our posteriors can
make sure the mean and peak of the denoised probability of measuring
$\ket{0}$ are around the ideal value $0.5$. The results in 
Figure~\ref{fig:testing-circ-H} indicate that when
applying Eq.~\eqref{eq:me-filter} to mitigate the measurement error, one has a
larger chance to obtain a denoised QoI close to the ideal value $0.5$ by using
the parameters inferred by our method. More importantly, the result by our
method is unbiased as the mean value of the denoised QoI is $0.5$.
\begin{figure}
     \centering
     \begin{subfigure}[b]{0.49\textwidth}
         \centering
         \includegraphics[width=\textwidth]{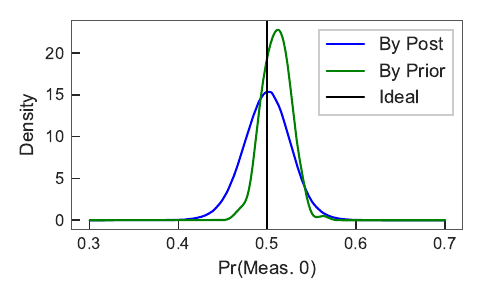}
         \caption{Qubit 1}
     \end{subfigure}
     \begin{subfigure}[b]{0.49\textwidth}
         \centering
         \includegraphics[width=\textwidth]{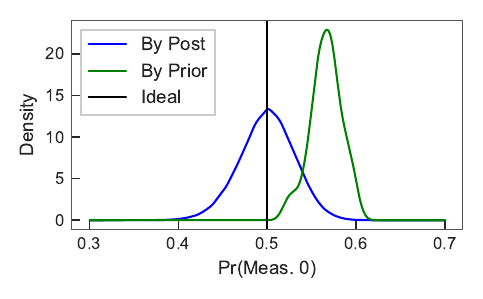}
         \caption{Qubit 2}
     \end{subfigure}
     \\
     \begin{subfigure}[b]{0.49\textwidth}
         \centering
         \includegraphics[width=\textwidth]{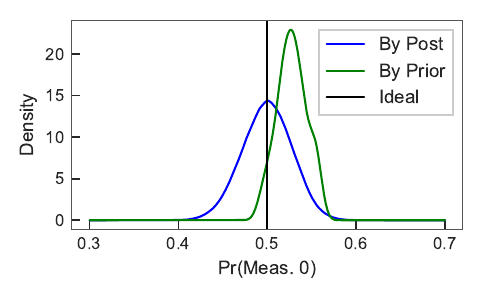}
         \caption{Qubit 3}
     \end{subfigure}
     \begin{subfigure}[b]{0.49\textwidth}
         \centering
         \includegraphics[width=\textwidth]{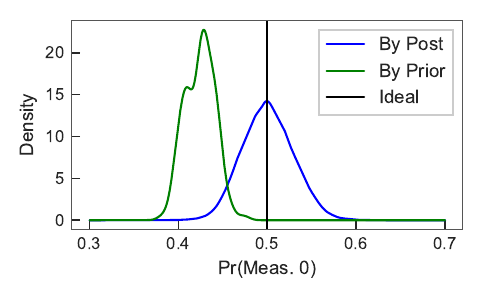}
         \caption{Qubit 4}
     \end{subfigure}
        \caption{%Gaussian KDE 
          PDFs of the probability of measuring $\ket{0}$ 
        denoised by vendor-provided parameters (priors) and by posteriors}
        \label{fig:DQoI}
\end{figure}

This test indicates that we can use the circuit shown in Figure~\ref{fig:test1}
to estimate measurement error in multiple qubits at the same time. It only
requires to prepare initial state using ground states, and the total number of
gates are linearly dependent on the number of qubits.

\subsubsection{Application on State Tomography}

After obtaining the error model parameters, we can further use this model
mitigate the measurement errors in other circuits.
We first apply error filters to the results of state tomography on circuits that
make bell basis from $\ket{00}$ and $\ket{000}$. Qubit 0 and 1 are used for
2-qubit state and Qubit 0 to 2 are used for 3-qubit state in \texttt{ibmqx2}.
The fidelity between density matrices from (corrected) state tomography result
and theoretical quantum state is listed in Table~\ref{tab:tomo}. For the 2-qubit
state tomography, filters constructed from posterior means by the consistent
Bayesian and the standard Bayesian provides similar fidelity as that by the
Qiskit filter. However, for 3-qubit tomography, filters from both
Bayesian methods yield better fidelity, and their performance are
similar. We note that the Qiskit filter assumes correlation in the 
measurements, which requires more model parameters, while our
model does not. The fidelity in Table~\ref{tab:tomo} indicates that Bayesian
methods enable us to use fewer parameters to obtain better results.
%provide better choices of parameter values in this experiment.
%-------------------------------------------------------------------------------
\begin{table}
    \caption{Fidelity of state tomography results filtered by various error filters}
    \label{tab:tomo}
    \begin{minipage}{\columnwidth}
    \begin{center}
    \begin{tabular}{lcccc} \toprule
    \multirow{2}{*}{State}                   & \multicolumn{4}{c}{Fidelity} \footnote{``Qiskit Method'' means to \texttt{CompleteMeasFitter} in Qiskit~\cite{Qiskit}. ``Cons. mean'' implies the transition matrix is created from posterior mean by Algorithm~\ref{alg:noise_params}. ``Stand. Mean'' means the transition matrix is created from posterior mean by standard Bayesian}                      \\ \cline{2-5} 
                                                 & Raw Data &By Qiskit Method &By Cons. Mean &By Stand. Mean \\ \midrule
    $\frac{1}{\sqrt{2}}(\ket{00} + \ket{11})$    & 0.9051   & 0.9800          & 0.9781     & 0.9783        \\ 
    $\frac{1}{\sqrt{2}}(\ket{01} + \ket{10})$    & 0.9157   & 0.9803          & 0.9806     & 0.9808        \\ 
    $\frac{1}{\sqrt{2}}(\ket{000} + \ket{111})$  & 0.7389   & 0.9227          & 0.9390     & 0.9391        \\ 
    $\frac{1}{\sqrt{2}}(\ket{010} + \ket{101})$  & 0.6719   & 0.8970          & 0.9203     & 0.9207        \\ 
    $\frac{1}{\sqrt{2}}(\ket{100} + \ket{011})$  & 0.7006   & 0.9121          & 0.9254     & 0.9207        \\ 
    $\frac{1}{\sqrt{2}}(\ket{110} + \ket{001})$  & 0.6974   & 0.8863          & 0.9443     & 0.9446        \\ \bottomrule
    \end{tabular}
    \end{center}
    \end{minipage}
\end{table}
%-------------------------------------------------------------------------------

\subsubsection{Application on Grover's Search and QAOA}

Next, we apply our filter on Grover's search and QAOA circuits from~\cite{algs}.
We measure Qubit $1$ and $2$ in \texttt{ibmqx2} for Grover's search circuit. The
exact solution of this Grover's search example is $\ket{11}$ and the theoretical
probability is $1$. Thus, in this case, we compare the probability of measuring
$\ket{11}$ by running in the real device \texttt{ibmqx2} and the denoised
probabilities from error filters based on Qiskit method
\texttt{CompleteMeasFitter}, QDT in~\cite{QDT2020}, mean and MAP of posteriors
from standard Bayesian, and mean and MAP of posteriors from the
consistent Bayesian. All circuits for the Qiskit filter and QDT filters are
executed for $8192$ shots, and each probability used in both filters is 
estimated from $8192$ measurement outcomes.

In addition, as we do not expect the quantum computer has a stable environment,
in order to see the robustness of each method in comparison, after the data for
creating error filters are collected, we run our Grover's search circuit at
several different time and then apply the same set of filters. All results are
listed in Table~\ref{tab:grover}.

\begin{table}
    \caption{Probability of measuring $\ket{11}$ in Grover's search example}
    \label{tab:grover}
    \begin{minipage}{\columnwidth}
    \begin{center}
    \begin{tabular}{lcccccc}\toprule
    Method/Source \footnote{``Qiskit Method'' means to \texttt{CompleteMeasFitter} in Qiskit~\cite{Qiskit}, QDT refers to filter in~\cite{QDT2020}, ``Stand.'' stands for Standard Bayesian, and ``Cons.'' refers to Algorithm~\ref{alg:noise_params}. MAP and mean represent the error filters are created from the MAP and mean of posteriors.}      & Hour 0 \footnote{``Hour X'' means the experiment is conducted X hours after the data for error filers of all listed methods are collected} & Hour 2 & Hour 4 & Hour 8 & Hour 12& Hour 16\\\midrule
    Raw Data     & 0.6727        & 0.6930 & 0.6724 & 0.6740 & 0.6917 & 0.6841 \\
    Qiskit Method   & 0.7097        & 0.7335 & 0.7104 & 0.7120 & 0.7323 & 0.7241 \\
    QDT             & 0.7107        & 0.7332 & 0.7087 & 0.7108 & 0.7305 & 0.7224 \\ 
    Stand. Mean         & 0.9099        & 0.9324 & 0.9063 & 0.9088 & 0.9290 & 0.9192 \\
    Stand. MAP        & 0.8378        & 0.8635 & 0.8372 & 0.8392 & 0.8616 & 0.8522 \\ 
    Cons. Mean        & 0.9128        & 0.9351 & 0.9088 & 0.9114 & 0.9316 & 0.9219 \\
    Cons. MAP       & 0.8920        & 0.9158 & 0.8914 & 0.8936 & 0.9128 & 0.9034 \\\bottomrule
    \end{tabular}
    \end{center}
    \end{minipage}
\end{table}

As shown in Table~\ref{tab:grover}, both Bayesian methods yield best performance
among all the methods while the filters constructed from posterior mean are
better than the filters constructed from posterior MAP. In all six time slots,
the mean and MAP from the consistent Bayesian provide sightly better denoising
effect than those from the standard Bayesian. 

The QAOA example includes two rounds and parameters for QAOA circuits are
set as $(\gamma_1, \beta_1) = (0.2\pi, 0.15\pi)$ and $(\gamma_2,
\beta_2) = (0.4\pi, 0.05\pi)$~\cite{algs}.  The graph of the QAOA example
in~\cite{algs} is shown in Figure~\ref{fig:QAOA-grah}, which has maximum
objective value $3$ in
Max-Cut problem and 6 bit-string optimal solution $\ket{0010}$, $\ket{0101}$,
$\ket{0110}$, $\ket{1001}$, $\ket{1010}$, $\ket{1101}$ (\texttt{ibmqx2} uses
little-endien convention, so the rightmost bit is Node 1 and the leftmost bit is
Node 4). Because the graph in Figure~\ref{fig:QAOA-grah} is a subgraph of the
coupling map of  \texttt{ibmqx2}, we map the nodes to qubits exactly.

\begin{figure}
\centering
% \begin{tikzpicture}
% \draw 
% (0,1) node[circle, black, draw](1){1}
% (0,0) node[circle, black, draw](2){2}
% (1,0) node[circle, black, draw](3){3}
% (0,-1) node[circle, black, draw](4){4};

% \draw[-] (1) -- (2);
% \draw[-] (3) -- (2);
% \draw[-] (4) -- (2);
% \draw[-] (3) -- (4);
% \end{tikzpicture}
\includegraphics[width=0.2\linewidth]{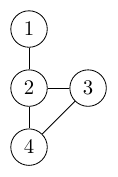}
\caption{The graph of QAOA example in~\cite{algs}} \label{fig:QAOA-grah}
\end{figure}

The average size of a cut and the probability of measuring an optimal solution are two quantities to compare in
this experiment. Moreover, the results from simulator is also provided as a
indicator for the situation without noise. The remaining procedures are the same
as those in Grover's search experiment. The data is reported in
Table~\ref{tab:QAOAM2} and~\ref{tab:QAOAProb}. Here, ``QDT'' error filter
from~\cite{QDT2020} is built under the assumption that measurement operations 
are independent between each qubit due to the large amount of testing circuits
that correlation assumption requires ($6^4$ circuits if assume qubits are
correlated in measurement).

\begin{table}
    \caption{Average size of a sampled cut in QAOA example}
    \label{tab:QAOAM2}
    \begin{center}
    \begin{tabular}{lcccccc} \toprule
    Method/Source & Hour 0 & Hour 2 & Hour 4 & Hour 8 & Hour 12 & Hour 16 \\ \midrule
    Simulator     & 2.8637 & 2.8642 & 2.8651 & 2.8652 & 2.8642  & 2.8626  \\ 
    Raw Data   & 2.3005 & 2.3579 & 2.3197 & 2.2823 & 2.3063  & 2.2871  \\ 
    Qiskit Method & 2.3783 & 2.4623 & 2.4247 & 2.3786 & 2.4063  & 2.3926  \\
    QDT           & 2.3812 & 2.4453 & 2.4016 & 2.3589 & 2.3851  & 2.3612  \\ 
    Stand. Mean       & 2.4878 & 2.5483 & 2.5059 & 2.4551 & 2.4860  & 2.4581  \\ 
    Stand. MAP       & 2.4080 & 2.4708 & 2.4222 & 2.3800 & 2.4059  & 2.3829  \\
    Cons. Mean      & 2.4911 & 2.5518 & 2.5089 & 2.4578 & 2.4891  & 2.4612  \\ 
    Cons. MAP      & 2.4407 & 2.4996 & 2.4554 & 2.4109 & 2.4382  & 2.4133  \\ \bottomrule
    \end{tabular}
    \end{center}
\end{table}

\begin{table}
    \caption{Probability of measuring an optimal solution in QAOA example}
    \label{tab:QAOAProb}
    \begin{center}
    \begin{tabular}{lcccccc} \toprule 
    Method/Source & Hour 0 & Hour 2 & Hour 4 & Hour 8 & Hour 12 & Hour 16 \\ \midrule
    Simulator     & 0.8930 & 0.8937 & 0.8941 & 0.8943 & 0.8940  & 0.8940  \\ 
    Raw Data   & 0.5784 & 0.6038 & 0.5895 & 0.5725 & 0.5748  & 0.5740  \\
    Qiskit Method & 0.5968 & 0.6456 & 0.6316 & 0.6074 & 0.6140  & 0.6155  \\ 
    QDT           & 0.6400 & 0.6698 & 0.6525 & 0.6312 & 0.6331  & 0.6325  \\ 
    Stand. Mean       & 0.6952 & 0.7239 & 0.7033 & 0.6766 & 0.6787  & 0.6797  \\ 
    Stand. MAP       & 0.6444 & 0.6695 & 0.6508 & 0.6305 & 0.6309  & 0.6317  \\
    Cons. Mean      & 0.6975 & 0.7265 & 0.7058 & 0.6790 & 0.6810  & 0.6822  \\ 
    Cons. MAP      & 0.6610 & 0.6860 & 0.6672 & 0.6431 & 0.6439  & 0.6452  \\ \bottomrule
    \end{tabular}
    \end{center}
\end{table}

The conclusion from Table~\ref{tab:QAOAM2} and~\ref{tab:QAOAProb} is basically
the same as that from Table~\ref{tab:grover}. Namely, Bayesian methods, especially
filters from posterior mean, outperform other methods, and parameters inferred
by the consistent Bayesian works slightly better than those by the standard
Bayesian in all six time slots. From both Grover's search and QAOA examples, we
can see the accuracy of both Bayesian approaches are better than the existing
methods.

%From both Grover's search and QAOA examples, we can see that treating noise parameters as random variables is more suited to modeling the error rates.

\subsubsection{Application on Random Clifford Circuits}

Finally, we test the measurement-error filtering for random 2-Qubit Clifford
circuits with 1, 2, 3, and 4 2-Qubit Clifford operators (i.e., length 1, 2, 3,
4). For each length, $16$ random circuits are generated to draw a boxplot and each
circuit is run for $8192$ shots. The results are shown in
Figure~\ref{fig:clifford}. While the theoretical output of 2-Qubit Clifford
circuit is $\ket{00}$ with probability 1, Figure~\ref{fig:clifford} demonstrate that the
filter constructed from posterior mean estimated by standard Bayesian provides
best performance. The consistent Bayesian results in almost the same
results as the standard Bayesian method.

\begin{figure}
     \centering
     \begin{subfigure}[b]{0.49\textwidth}
         \centering
         \includegraphics[width=\textwidth]{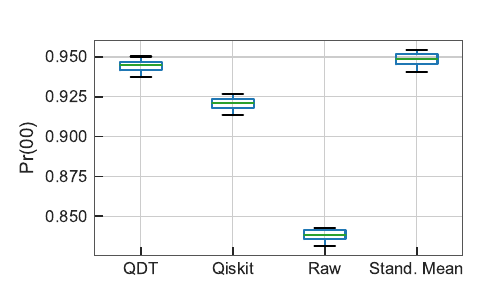}
         \caption{Length 1}
     \end{subfigure}
     \begin{subfigure}[b]{0.49\textwidth}
         \centering
         \includegraphics[width=\textwidth]{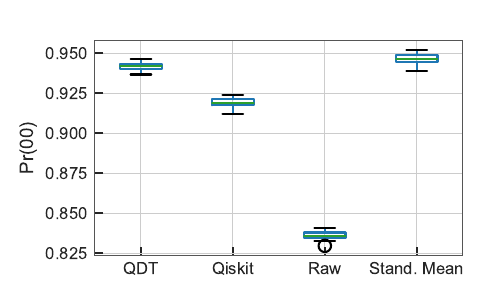}
         \caption{Length 2}
     \end{subfigure}\\
     \begin{subfigure}[b]{0.49\textwidth}
         \centering
         \includegraphics[width=\textwidth]{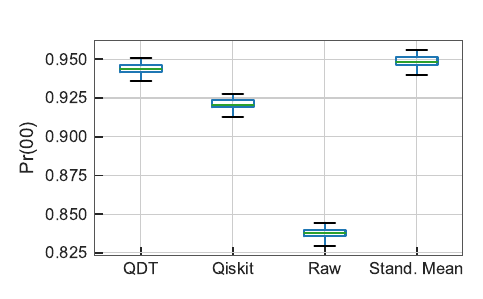}
         \caption{Length 3}
     \end{subfigure}
     \begin{subfigure}[b]{0.49\textwidth}
         \centering
         \includegraphics[width=\textwidth]{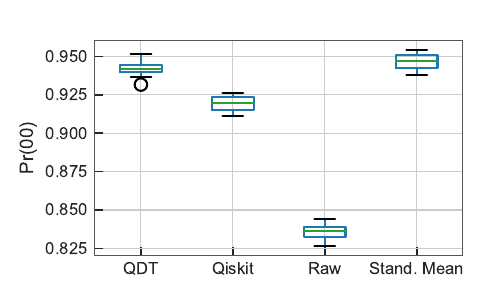}
         \caption{Length 4}
     \end{subfigure}
        \caption{Measurement-error filtering for random 2-Qubit Clifford
        circuits. ``Length'' represents the number of Clifford operators in the
        circuit. ``Probability'' means the probability of measuring $\ket{00}$.}
        \label{fig:clifford}
\end{figure}

\subsection{Gate and Measurement Error Filtering Experiment}\label{sec:gate-exp}

%As our gate error model Eq.~\eqref{eq:pm} is limited to bit-flip error on some
%simple circuit, 
We consider the circuit with $200$ NOT gates as shown in Figure~\ref{fig:200X}.
We still use machine \texttt{ibmqx2} and run the experiment twice separately on
Qubit 1 and Qubit 2. In each trial, the circuit is executed $1024 \times 128$
times where readouts from every $1024$ runs are used to estimate the
QoI, i.e., the probability of measuring $\ket{0}$. Namely, we collect $128$
samples of the QoI.
%Note that the gate error model \eqref{eq:pm} require us to know the noise-free probability distribution of all basis if we want to infer $\eg$. So we use experiment in two ways: assume we know noise-free probability distribution then infer $(\ez,\eo,\eg)$ and assume we know $(\ez,\eo,\eg)$ and do denoising.
%-------------------------------------------------------------------------------
\begin{figure}
    \centering
    % \[
    % \Qcircuit @C=1em @R=.7em {
    %     &\lstick{\ket{0}} &\gate{X} &\qw &\cdots & &\gate{X} &\meter \gategroup{1}{3}{2}{7}{0.1em}{_\}} \\
    %     &&&&&&&\\
    %     & & & &\mbox{$200$ times} &&&
    % }
    % \]
    \includegraphics[width=0.2\linewidth]{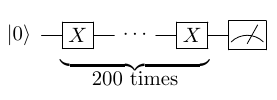}
    \caption{Experiment circuit for gate and measurement error
    mitigation.}
    \label{fig:200X}
\end{figure}
%-------------------------------------------------------------------------------
Because the aforementioned Qiskit method \texttt{CompleteMeasFitter} and QDT are
for measurement errors, in this section, we only compare the results from
standard Bayesian and the consistent Bayesian with the same priors and dataset. 
%Note that our QoI is still the probability of measuring 0 when we use Algorithm~\ref{alg:noise_params} to create filters. 
The priors are truncated normal $N(\ezi^0,0.1^2)$, $N(\eoi^0,0.1^2)$, and
% $N(\egi^0,0.01^2)$ with range $(0,1)$. Again, $\ezi^0, \eoi^0, \egi^0$ are
$N(\egi^0,0.005^2)$ with range $(0,1)$. Again, $\ezi^0, \eoi^0, \egi^0$ are
vendor-provided values from IBM's daily calibration and standard deviations are
chosen to make prior relatively flat due to the lack of knowledge on these
parameters. 

\subsubsection{Inference for Noise Parameters}

%Figure~\ref{fig:QoI-gate} illustrates that the consistent Bayesian 
%still yields posteriors of QoI $\pi_{\mathcal{D}}^{Q(\text{post})}$ that matches
%the data $\pi_{\mathcal{D}}^{\text{obs}}$ quite well.
Figure~\ref{fig:eg-gate} shows the distribution of $\eg$ in Qubit 1 and 2. Both
distributions are right-skewed. Table~\ref{tab:stats-gate} provides the
numerical values for mean and MAP. In Table~\ref{tab:stats-gate}, we can see
both methods give similar measurement error parameter $\ezi$ and $\eoi$ on Qubit
1 and 2, but the gate error rate $\eg$ are not always similar. More importantly,
as shown in Figure~\ref{fig:sbob-gate}, posteriors of the QoI from the 
consistent Bayesianp, i.e., $\pi_{\mathcal{D}}^{Q(\text{post})}$
matches the distribution of data, i.e., $\pi_{\mathcal{D}}^{\text{obs}}$ quite
well. On the other hand, the posterior distribution of the QoI generated by
posteriors distributions of model parameters from the standard Bayesian can 
match the empirial mean of the data only while the shape of the PDF is quite
different. 
% in mean or MAP (as our distribution of data is quite symmetric).

%\begin{figure}[ht!]
%     \centering
%     \begin{subfigure}[b]{0.49\textwidth}
%         \centering
%         \includegraphics[width=\textwidth]{Pics/Gate/QoI-Qubit1.pdf}
%         \caption{Qubit 1}
%     \end{subfigure}
%     \begin{subfigure}[b]{0.49\textwidth}
%         \centering
%         \includegraphics[width=\textwidth]{Pics/Gate/QoI-Qubit2.pdf}
%         \caption{Qubit 2}
%     \end{subfigure}
%        \caption{Gaussian KDE of probability of measuring 0 from data and from values generated by posteriors for circuit in Figure~\ref{fig:200X}.}
%        \label{fig:QoI-gate}
%\end{figure}

\begin{figure}
     \centering
     \begin{subfigure}[b]{0.49\textwidth}
         \centering
         \includegraphics[width=\textwidth]{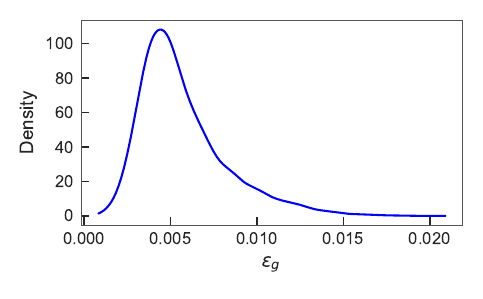}
         \caption{Qubit 1}
     \end{subfigure}
     \begin{subfigure}[b]{0.49\textwidth}
         \centering
         \includegraphics[width=\textwidth]{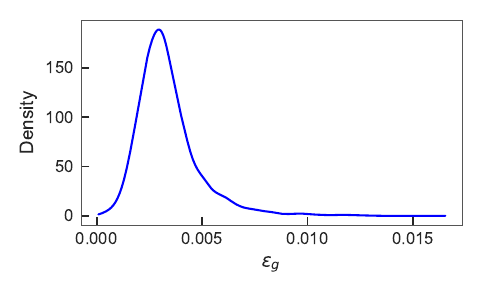}
         \caption{Qubit 2}
     \end{subfigure}
        \caption{Gaussian KDE of $\eg$}
        \label{fig:eg-gate}
\end{figure}

% \begin{table}[ht!]
%     \centering
%     \footnotesize
%     \caption{Measurement error parameters of the consistent and standard
%     Bayesian inference.}
%     \label{tab:stats-gate}
%     \begin{tabular}{|r|c|c|c|c|} \hline
%          &Qubit 1 &Qubit 2\\ \hline
% %    Consistent Bayesian Acceptance Rate &$15.6\%$ &$11.4\%$ \\ \hline
% %    Consistent KL-div($\pi_\Dcal^{Q(\text{posterior})},\pi_\Dcal^{\text{obs}}$) &$0.001381$ &$0.000844$ \\ \hline
%     Consistent Post. Mean $(1-\ezi,1-\eoi, \eg)$ &(0.9255, 0.8922, 0.009868) &(0.9229, 0.8856, 0.007609)  \\ 
%     Consistent Post. MAP $(1-\ezi,1-\eoi, \eg)$ &(0.9756, 0.8837, 0.009654) &(0.9770, 0.9485, 0.006582)  \\ \hline
%     Standard Post. Mean $(1-\ezi,1-\eoi, \eg)$  &(0.9221, 0.8939, 0.009367) &(0.9214, 0.8871, 0.005965)  \\ 
%     Standard Post. MAP $(1-\ezi,1-\eoi, \eg)$  &(0.9758, 0.8835, 0.013101) &(0.9836, 0.9354, 0.006906)  \\ \hline
%     \end{tabular}
% \end{table}
\begin{table} %Muqing: Since we now make the 0.5*old_epsilon_g = new_epsilon_g, we need to multiply 0.5 to all gate error rate here.
    \caption{Measurement error parameters of the consistent and standard
    Bayesian inference.}
    \label{tab:stats-gate}
    \begin{center}
    \begin{tabular}{rcccc} \toprule
         &Qubit 1 &Qubit 2\\ \midrule
    Consistent Post. Mean $(1-\ezi,1-\eoi, \eg)$ &(0.9255, 0.8922, 0.004934) &(0.9229, 0.8856, 0.003804)  \\ 
    Consistent Post. MAP $(1-\ezi,1-\eoi, \eg)$ &(0.9756, 0.8837, 0.004827) &(0.9770, 0.9485, 0.003291)  \\ 
    Standard Post. Mean $(1-\ezi,1-\eoi, \eg)$  &(0.9221, 0.8939, 0.004683) &(0.9214, 0.8871, 0.002982)  \\ 
    Standard Post. MAP $(1-\ezi,1-\eoi, \eg)$  &(0.9758, 0.8835, 0.006550) &(0.9836, 0.9354, 0.003453)  \\ \bottomrule
    \end{tabular}
    \end{center}
\end{table}

\begin{figure}
     \centering
     \begin{subfigure}[b]{0.49\textwidth}
         \centering
         \includegraphics[width=\textwidth]{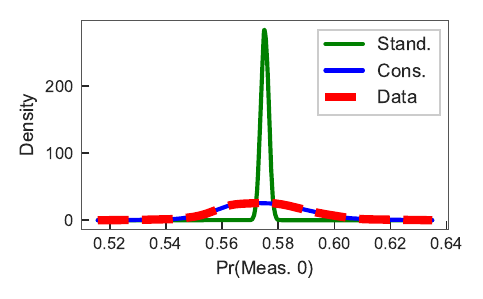}
         \caption{Qubit 1}
     \end{subfigure}
     \begin{subfigure}[b]{0.49\textwidth}
         \centering
         \includegraphics[width=\textwidth]{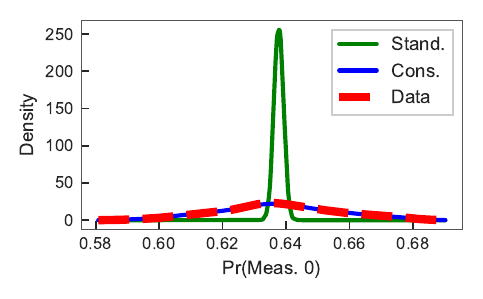}
         \caption{Qubit 2}
     \end{subfigure}
        \caption{Posterior distribution of the QoI, i.e.,
        $\pi_{\mathcal{D}}^{Q(\text{post})}$, generated by the posterior
        distribution of model parameters from both Bayesian methods}
        \label{fig:sbob-gate}
\end{figure}

%\newpage
\subsubsection{Error Filtering}
Using the posterior means from Table~\ref{tab:stats-gate}, we construct gate and
measurement error filters and apply them on the 128 samples of the QoI.
(i.e., probabilities of measuring 0) on Qubit 1 and Qubit 2. The results are
displayed in Figure~\ref{fig:gate-denoised}. 
%In the most of time, the filter created by parameters from Algorithm~\ref{alg:noise_params} can recover the ideal probability $1$. 
It shows that both Bayesian approaches we use can recover the exact value $1$
with high probability. More importantly, the consistent Bayesian outperforms the
standard one as it recovers the exact value $1$ with larger chance.
especially in the test on Qubit 2.
%Again, we can also see filters from Algorithm~\ref{alg:noise_params} are clearly better those from standard Bayesian, especially for Qubit 2.

\begin{figure}
    \centering
     \begin{subfigure}[b]{0.49\textwidth}
         \centering
         \includegraphics[width=\textwidth]{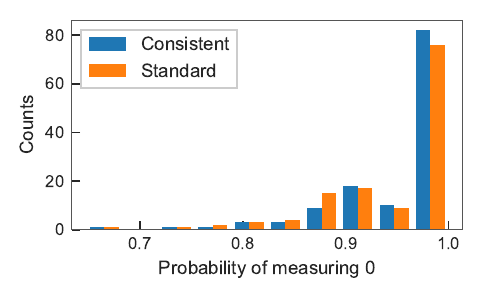}
         \caption{Qubit 1}
     \end{subfigure}
     \begin{subfigure}[b]{0.49\textwidth}
         \centering
         \includegraphics[width=\textwidth]{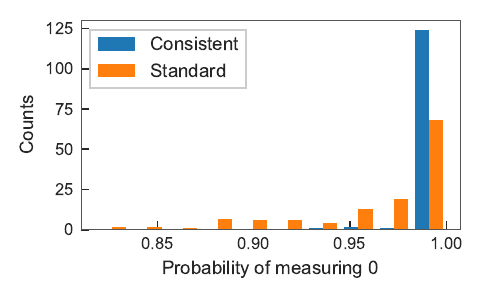}
         \caption{Qubit 2}
     \end{subfigure}
    \caption{Denoised (both gate and measurement) probability of measuring 0. Parameters used are posterior mean from Table~\ref{tab:stats-gate}.}
    \label{fig:gate-denoised}
\end{figure}

The data for error filters in Section~\ref{sec:meas-exp} and experiment in
Section~\ref{sec:gate-exp} were collected within one hour, so it is reasonable
to use posteriors in Section~\ref{sec:gate-exp} to denoise the Grover's search
data in Section~\ref{sec:meas-exp}. However, comparing the values of measurement
error parameters in Table~\ref{tab:QoI} and in Table~\ref{tab:stats-gate}, we can see there are some noticeable differences. 
Table~\ref{tab:mix} provides the results of using parameters in
Section~\ref{sec:gate-exp} to filter out errors in data used in 
Section~\ref{sec:meas-exp}. 
%Comparing those values to the ones in Table~\ref{tab:grover}, 
We can see they are better than Qiskit method and QDT, but worse than values 
from either Bayesian methods shown in Tabel~\ref{tab:grover}.

\begin{table}
    \caption{Probability of measuring $\ket{11}$ in Grover's search Example denoised by parameters in Section~\ref{sec:gate-exp}}
    \label{tab:mix}
    \begin{center}
    \begin{tabular}{lcccccc} \toprule
    Method   & Hour 0 & Hour 2 & Hour 4 & Hour 8 & Hour 12 & Hour 16 \\ \midrule
    Stand. Mean  & 0.8398 & 0.8680 & 0.8392 & 0.8414 & 0.8656  & 0.8561  \\ 
    Stand. MAP  & 0.8116 & 0.8367 & 0.8110 & 0.8131 & 0.8348  & 0.8257  \\ 
    Cons. Mean & 0.8434 & 0.8716 & 0.8428 & 0.8450 & 0.8691  & 0.8595  \\
    Cons. MAP & 0.7992 & 0.8240 & 0.7986 & 0.8005 & 0.8221  & 0.8131  \\ \bottomrule
    \end{tabular}
    \end{center}
\end{table}

One possible explanation is, with $200$ gates, our model is much more sensitive
to the gate error than the measurement error. A comparison of the sensitivity is
shown in Figure~\ref{fig:sens}, where the results are obtained by using the
error models Eqs.~\eqref{eq:me} and \eqref{eq:pm}. In 
Figure~\ref{fig:sens} (a) and (b), we can see
that when $\epsilon_g$ is fixed, the QoI changes linearly and slowly as $\ezi$
or $\eoi$ varies. However, as shown in Figure~\ref{fig:sens} (c) when $\ezi$ and
$\eoi$ are fixed, the QoI changes rapidly as $\epsilon_g$ increases.
%In this case, $\egi$ is estimated more precisely than $\ezi, \eoi$
%during the experiment in Section~\ref{sec:gate-exp}. On the other hand, 
the estimation of measurement error in Section~\ref{sec:meas-exp} uses circuits
that have 0.5 chance to measure either $\ket{0}$ or $\ket{1}$ without noise and
this distribution does not change when a Hadamard gate suffers from bit-flip or phase-flip error,
%In this case, $\ezi$ and $\eoi$ should be estimated equal-precisely in Section~\ref{sec:meas-exp} and 
so it yields a better performance. 
%than parameters from Section~\ref{sec:gate-exp}.

\begin{figure}
     \centering
     \begin{subfigure}[b]{0.32\textwidth}
         \centering
         \includegraphics[width=\textwidth]{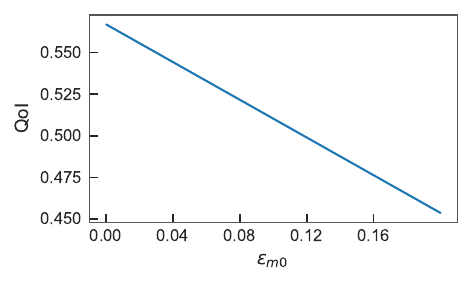}
        %  \caption{$\eo = 0, \eg = 0.01$}
        \caption{$\eo = 0, \eg = 0.005$}
     \end{subfigure}
     \begin{subfigure}[b]{0.32\textwidth}
         \centering
         \includegraphics[width=\textwidth]{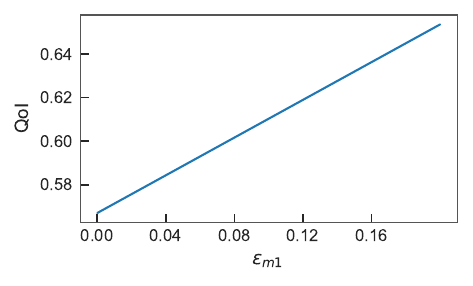}
        %  \caption{$\ez = 0, \eg = 0.01$}
        \caption{$\eo = 0, \eg = 0.005$}
     \end{subfigure}
     \begin{subfigure}[b]{0.32\textwidth}
         \centering
         \includegraphics[width=\textwidth]{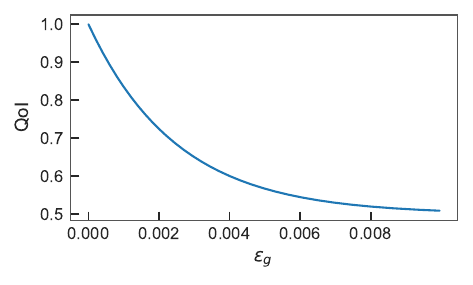}
         \caption{$\eo = 0, \ez = 0$}
     \end{subfigure}
        \caption{Sensitivity analysis of forward model $Q$ with 200 NOT gate.}
        \label{fig:sens}
\end{figure}

%\newpage

\section{Discussion and Future Works}

In this work, we extend a bit-flip error model from a single gate case to
multiple gate case, and provide theoretical analysis to prove the existence of
the error mitigation solution for both cases.
In some noise models, such as depolarizing error model, the rate of bit-flip error is associated with the rates of other types of errors~\cite[p. 379]{NC2010}. Thus, the inference of bit-flip error rates could provide a connection to a more general noise model.
We propose to use Bayesian approaches to infer parameters in the error models to characterize the propagation of the device noise in QC algorithms more
effectively. The experiments in Section~\ref{sec:exp} demonstrate that our methodology
outperforms two existing methods on the same error models over a wide range of time, while the number of testing circuits is
linear or constant to the number of qubits. 
The consistent Bayesian approach is, in general, better than the standard Bayesian.
These results indicate that our error models can characterize the device noise quite well, and they help to
understand the propagation of such noise in QC algorithms. 
%Consequently, the error mitigation approaches based on these error models yield good results in our tests.

There are still several limitations in our methodology. One issue that affects
the scalability of our method is the exponentially large matrix in the denoising
step. The dimension of matrices can be reduced if we can identify the qubits that
are independent during the measurement step of an algorithm and filter their
measurement outcome separately. 
A recent work in~\cite{space_red_meas} also indicates a scheme to reduce the dimension of the transition matrix by limiting the range of bases that are put into consideration.
% To solve this problem, one can exploit the tensor
% product structure of the linear system used in the error filtering step to
% perform this task in a parallel way~\cite{tensorform}.
Also, a parallel algorithm proposed in~\cite{tensorform} can exploit the tensor-product structure of the linear system in the error filtering step to speedup the calculation.
On the other hand, 
because the method of estimating the distribution of model parameters is not
limited by the two models we discussed in the paper, a consideration for pairwise-correlated measurement error
model discussed in~\cite{Bravyi2020} probably be helpful for inferring 
correlated measurement error rates.
% Also, our gate error model can only deal with bit-flip error in some simple
% circuits. More efforts are required to further improve our error mitigation
% model as such to make it more general and accurate. 
%As for the gate error model, its applicable gate and error types are limited. More efforts are required to further improve our error mitigation model as such to make it more general and accurate. 
A potential extension to the applicable gates is to modify the model to accommodate multi-qubit bit-flip error instead of individual-qubit error. This is because more gates commute with $X^{\otimes n}$ than with elements in $\{X,I\}^{\otimes n}$ for $n \geq 2$. For example, $X\otimes X$ commute with matrix $A \otimes B$ and $e^{-i \delta A\otimes B}$ for $(A,B) \in \{Y,Z\}^{\otimes 2}\cup \{I,X\}^{\otimes 2}$ and arbitrary $\delta$, where the form  $e^{-i \delta A\otimes B}$ is generally utilized in quantum simulation~\cite{Tacchino2019}.

% This work is supported by Defense Advanced Research Projects Agency as part of
% ``The Quantum Computing Revolution and Optimization: Challenges and
% Opportunities'' project.

\section*{Acknowledgement}

% This work is supported by Defense Advanced Research Projects Agency as part of the project W911NF2010022: {\em The Quantum
% Computing Revolution and Optimization: Challenges and Opportunities}.
% Ang Li is work is supported by the U.S. Department of Energy, Office of Science, National Quantum Information Science Research Centers, Co-design Center for Quantum Advantage (C2QA) under contract number DE-SC0012704

%\grantsponsor{⟨sponsorID⟩}{⟨name⟩}{⟨url⟩}
%\grantnum[⟨url⟩]{⟨sponsorID⟩}{⟨number⟩}
%At present {⟨sponsorID⟩} is chosen by the authors and can be an arbitrary key in the same way the label of a \cite is arbitrarily chosen. 

This works was partially supported by Defense Advanced Research Projects Agency under Grant No.:W911NF2010022. Xiu Yang was also supported by NSF CAREER DMS-2143915. Muqing Zheng was also supported by U.S. Department of Energy, Office of Science, National Quantum Information Science Research Centers, Quantum Science Center (QSC). Ang Li's work was supported by the U.S. Department of Energy, Office of Science, National Quantum Information Science Research Centers, Co-design Center for Quantum Advantage (C$^2$QA) under contract number DE-SC0012704. The Pacific Northwest National Laboratory is operated by the U.S. Department of Energy under Contract DE-AC05-76RL01830.
This research used resources of the Oak Ridge Leadership Computing
Facility, which is a DOE Office of Science User Facility supported under
Contract DE-AC05-00OR22725.
This research also used resources of the National Energy Research Scientific Computing Center (NERSC), a U.S. Department of Energy Office of Science User Facility located at Lawrence Berkeley National Laboratory, operated under Contract No. DE-AC02-05CH11231 using NERSC award ERCAP0022228.

%Ang Li is supported by the XCITe crosscut of Co-design Center for Quantum Advantage (C2QA), a National Quantum Information Science Research Center of the U.S. Department of Energy (DOE). The Pacific Northwest National Laboratory is operated by Battelle for the U.S. Department of Energy under contract DE-AC05-76RL01830.

\bibliographystyle{unsrt}
\bibliography{refs}

\end{document}